\documentclass{article}
\usepackage[utf8]{inputenc}
\usepackage{amsmath, amsthm, amssymb}
\usepackage{xcolor}
\usepackage[round]{natbib}
\usepackage{paavomacros}
\usepackage[a4paper, total={6in, 8in}]{geometry}
\usepackage{graphicx}
\usepackage{placeins}

\newtheorem{assumption}{Assumption}
\newtheorem{lem}{Lemma}

\newtheorem{example}{Example}

\newcommand{\rnc}{\renewcommand}
\newcommand{\nc}{\newcommand}
\newcommand{\mrm}{\mathrm}

\nc{\mb}{\mathbb}
\rnc{\mc}{\mathcal}
\rnc{\E}{\mb{E}}
\rnc{\N}{\mb{N}}
\rnc{\R}{\mb{R}}
\nc{\Q}{\mb{Q}}
\nc{\Qn}{\mb{Q}_n}
\rnc{\P}{\mb P}
\nc{\Pn}{\P_n}
\nc{\RP}{\tilde{\P}}
\nc{\RPn}{\tilde{\P}_n}
\nc{\PnG}{\Pn^{\mc G}}
\nc{\PG}{\P^{\mc G}}
\rnc{\d}{\mrm d}
\nc{\C}{\mc{C}}
\nc{\D}{\mc{D}}
\nc{\B}{\mc{B}}
\nc{\oPo}{\stackrel{\mrm p}{\rightarrow}}
\nc{\oWo}{\stackrel{w}{\rightarrow}}
\nc{\oDo}{\rightsquigarrow}
\nc{\eff}{\|F\|}
\nc{\cred}{\color{red}}
\nc{\tcb}{\textcolor{black}}
\nc{\tcr}{\textcolor{red}}
\DeclareMathOperator{\im}{im}

\title{Quasi-randomization-based inference for multivariate Mann-Whitney effects under random missingness}
\author{Dennis Dobler\footnote{Institute of Statistics, RWTH Aachen University, Aachen, Germany}, 
Jörg-Tobias Kuhn\footnote{Department of Rehabilitation Sciences, TU Dortmund University, 44227 Dortmund, Germany}, 
Lubna Amro\footnote{Department of Statistics, TU Dortmund University, 44227 Dortmund, Germany}, and Paavo Sattler$^{\ddagger,*}$}
\date{\today}
\newcommand{\bqan}{\begin{eqnarray}}
	\newcommand{\eqan}{\end{eqnarray}}
	
\newtheorem{cor}{Corollary}	
\newtheorem{thm}{Theorem}	
\newtheorem{rem}{Remark}	
\newtheorem{prop}{Proposition}

\begin{document}

\maketitle

\begin{abstract}
Marginal Mann-Whitney effects are widely used across various fields of research, and extensions of this estimand have been developed in many directions in statistical methodology. In this paper, we focus on an extensions for repeated measurements and factorial designs subject to randomly missing data. 
In a previous work by \cite{rubarth22}, asymptotically correct tests were developed under the assumption of deterministic missing indicators. In contrast, the approach in the present paper accounts for the stochastic nature of missing values under realistic mechanisms.
Thus,  the involved covariance matrix incorporates the true variability of missing data. The combination with a quasi-randomization procedure using random permutations within each data point yields asymptotically exact tests and a generally improved type-I error control.
Additionally, the tests control the type-I error for finite sample sizes in the special case of exchangeable sampling distributions.

Simulations across a wide range of settings demonstrate the benefits of the proposed method in small samples, also for different missingness mechanisms.
A real data analysis about school children learning math illustrates several practical aspects of the tests' application.
\end{abstract}
\textbf{Keywords:} asymptotic statistics, exact tests, factorial designs, missing data, relative effects, repeated measurements, resampling, robust statistics.

\section{Introduction}

Mann-Whitney effects, also known as `relative effects', are popular and easy-to-interpret parameters for comparing distributions.
They enable more robust statistical approaches than $t$-tests.
Originally introduced by \cite{Mann1947} in two-sample comparisons,
they are additionally applicable to ordinal data. Generalizations to multiple samples or multivariate settings with repeated measures were for example investigated by \cite{brunner2017rank}, \cite{dobler2020b},  and \cite{thiel2025}.
However, these studies assume complete observability---for many practical applications an unrealistic restriction.
Missing values frequently arise due to non-response, dropouts, or measurement limitations in applied research \citep{Schafer1997,little20}.

Methods for dealing with missing values have been developed for the analysis of means in various settings, e.g., \cite{Bhoj1978, Shao96, looney2003method, Amro2017, amro2024wildbootstrap}. 
These approaches address different scenarios, such as paired data, repeated measures, or factorial designs, and are tailored to handle missing observations in a statistically valid manner. In \cite{rubarth22}, tests for the multivariate Mann-Whitney effect were proposed specifically to handle missing values in a particular setting. 
This was extended in \cite{rubarth22a} to the multiple group case and further generalized to clustered data settings, where individual observations may be recorded repeatedly.
Both of these approaches primarily relied on distributional approximations, for example with F‑approximations, instead of  asymptotically exact testing.
Also, deterministic rather than random missingness schemes were considered.
These assumptions and approximations might be too unrealistic or simplistic in many fields of application.
That is why, in the present work, we develop statistical methodology which allows for randomly missing data and tests that are asymptotically exact.

To this end, we investigate alternative inferential approaches and analyze their statistical large-sample properties. 
As a means to improve the type-I error control, especially for small samples, we employ a resampling technique based on randomly re-arranging the entries of each data point, thereby creating an artificial regime of exchangeability.

In accordance with \cite{zhang23}, we classify the present approach as a \emph{quasi-randomiza\-tion test}, since the randomization is not based on a physical randomization mechanism in the experiment. The procedure used in the present paper is based on randomly transforming the data points; see e.g.\ \cite{romano89}. Such tests exhibit the potential for asymptotic guarantees \citep{Hoeffding52,romano90}, while controlling the type-I error probability for finite sample sizes in special cases of certain sampling distributional invariance, with related procedures being referred to as \emph{permutation tests} in parts of the more recent literature \citep{hemerik18,Goemann2021,hemerik24}.

The test statistics employed in the present paper are inspired by the Wald-type tests of \cite{rubarth22} and \cite{rubarth22a}.
Tests about multivariate Mann-Whitney effects in repeated measures designs under complete observability and for hypotheses formulated in terms of distribution functions were earlier developed by \cite{konietschke10}.
For the approach to test hypotheses in terms of multiple Mann-Whitney effects instead of multiple cumulative distribution functions (in a multiple samples framework), we refer to \cite{brunner2017rank}.
Parts of the methodological setting considered in the present work were also considered in the PhD thesis of \cite{jerdack87} who developed quadratic-form and rank-based tests for interchangeability, also under random missingness of some data.

Giving an overview, it is the aim of the present paper to develop statistical methods for combining
\begin{itemize}
    \item possibilities to treat randomly missing data in one or multiple components;
    \item a probabilistic estimand-based focus which avoids common paradoxes;
    \item analyses of the large sample properties;
    \item improving the reliability of inference methods by means of quasi-randomization while
    \item ensuring the finite type-I error control of tests under exchangeability;
    \item additionally tailoring the quasi-randomization technique to guarantee finite type-I error control under the weaker interchangeability, e.g., in testing for interaction effects in a general factorial design framework.
\end{itemize}

The paper is structured as follows:
In Section~\ref{Secmod}, we introduce the statistical model and its corresponding estimators, and we examine their asymptotic behavior under appropriate assumptions.
Next, Section~\ref{sec:rand_inference} presents the framework for quasi-randomization-based inference.
In Section~\ref{simulation}, we investigate the finite-sample performance of the proposed methods through a comprehensive simulation study, focusing on control of the type-I error probability and power under various alternatives.
Section~\ref{application} returns to the example from the beginning of the paper and provides an illustrative application based on the empirical data on children learning mathematics.
Finally, Section~\ref{conclusion} offers a summary and discusses potential directions for future research.
All proofs are included in an appendix, which also contains additional details and further simulation results.

  \section{Model, estimators, and their asymptotic behavior} \label{Secmod}
  \subsection{Model and estimands}
  \textcolor{black}{Let $(\Omega, \mathcal{A},P)$ be the underlying probability space that we will consider throughout the paper. We}
consider the following nonparametric repeated measures model with $n$ independent and identically distributed (i.i.d.) random vectors
\begin{align*}
    (\boldsymbol{X}_k, \boldsymbol{\lambda}_k)=& \ (X_{1k}, ... X_{dk},\lambda_{1k}, \dots, \lambda_{dk})^\top : \Omega \to \mathbb{R}^{2d}, \text{ with}\\ \lambda_{ik}= & \begin{cases}
	 1, \quad \text{if $X_{ik}$ is observed}	\\
	 	 0, \quad \text{if $X_{ik}$ is non-observed}
 \end{cases}
 i=1, ...,d; \ \ k=1, ...,n;
\end{align*} 
 of dimension $2d \geq 4$. For the most part, we assume the missingness mechanism `missing completely at random' (MCAR):
 \begin{assumption}
 \label{ass:MCAR}
  $(\boldsymbol{X}_1, \dots, \boldsymbol{X}_n )$ and $(\boldsymbol{\lambda}_1, \dots, \boldsymbol{\lambda}_n )$ are stochastically independent.
 \end{assumption}
 Under this formulation, the missingness indicators corresponding to the individual components may be mutually dependent. 
 Such dependence among the $\lambda$'s might naturally arise when strongly correlated components of the observation vector lead to a corresponding dependence in the missingness mechanism.

 The stochastic independence of the missing indicators from the values of the vector that the analysis will be conducted conditionally on the indicators and treated as fixed. However, we will also study MAR mechanisms (`missing at random')  through simulations in the appendix. 
 This framework allows for some influence of the observed values of $X_{ik}$ on the missingness mechanism \citep{little20}.
 
 We denote the normalized marginal cumulative distribution functions by
  $$ 
 X_{ik} \sim F_{i}(x)=\frac{1}{2}[F^+_{i}(x)+F^-_{i}(x)], \ i=1, ...,d, \ k=1, ...,n,
 $$ 
 where $F^+_{i}(x)=P(X_{i1}\leq x)$ and $F^-_{i}(x)=P(X_{i1}<x)$ are their right- and left-continuous versions, respectively. The normalized versions $F_{i}(x)$ and the corresponding estimators will allow for a straightforward inclusion of ties in the data \citep{Ruymgaart80, brunner18}. Thus, we require no assumptions regarding the discreteness or continuity of data generating mechanism. 
By $N=\sum_{i=1}^{d}\sum_{k=1}^{n}\lambda_{ik} =: \sum_{i=1}^d \lambda_{i\cdot}$ we denote the total number of observations.
 
Although the model considered in the present paper is nonparametric, there still exist different approaches for quantifying the stochastic superiority in the presence of multivariate data: for instance, \cite{brunner1999rank} and \cite{domhof2002rank} considered the weighted  marginal Mann-Whitney effects,
$ 
 q_{i}=\int{G_N(x)dF_{i}(x)},
 $ 
  where $G_N(x)=N^{-1}\sum_{i=1}^{d}\lambda_{i\cdot}F_{i}(x)$ is the weighted average of all distribution functions in the experiment. However, the use of $q_{i}$ as a measure of the causal differences between the distributions might not always be appropriate because it strongly depends on the observability of the data. 
  It is therefore not constant in the model and testing hypotheses based on the weighed Mann-Whitney effect might be questionable. To address this issue, we follow \cite{rubarth22} and \cite{brunner2017rank} and use unweighted Mann-Whitney effects as follows: 
  $$
 p_{i}=\int{H(x)dF_{i}(x)} = P(Z \preceq X_{i1}) := P(Z < X_{i1}) + \tfrac12 P(Z = X_{i1}).
 $$
  Here, the random variable $Z\sim  H$ is stochastically independent of $\vX_1, \dots, \vX_n$, and 
  $ H(x)=d^{-1}\sum_{i=1}^{d}F_{i}(x)$ is the unweighted average of all distribution functions. The Mann-Whitney effect $ p_{i}$ models the relative relationship of the distribution $F_{i}$ to the average distribution $H$. If $p_{i}<p_{j}$, the observations from distribution $F_{i}$ tend to be smaller than the observations from distribution $F_{j}$ \textit{relative to the average distribution $H$}, and vice versa if  $p_{i}>p_{j}$. In the case of $p_i = p_j$, none of the observations tends to result in smaller or larger values, even though the underlying distributions may differ considerably and a direct comparison of $F_i$ and $F_j$ might suggest otherwise.

{Denote by $\mathbb{R}^d_1 = \{\vx  = (x_1,\dots, x_d)^\top \in \mathbb R^d : x_1, \dots, x_d \geq 0,  \sum_{i=1}^d x_i =1 \}$ the unit simplex. To formulate null hypotheses for testing in this nonparametric setup,  let $\boldsymbol{p}=(p_{1}, ...,p_{d})^\top \in \mathbb{R}^d_1$ denote the true vector of the marginal Mann-Whitney effects.
 Also, let $\boldsymbol{C} \in \mathbb{R}^{r \times d}$ denote some non-trivial hypothesis matrix.
Null hypotheses take the form $H_0 = H_0(\vC, \vc) :\{\vC\boldsymbol{p}=\boldsymbol{c}\}$, which are tested against alternative hypotheses  $H_1 = H_1(\vC, \vc):\{\vC\boldsymbol{p}\neq \boldsymbol{c}\}$, for some $\boldsymbol{c} \in \mathbb{R}^r$.
For instance, if $\boldsymbol C= \boldsymbol I_d \in \mathbb{R}^{d \times d}$, the identity matrix, then any fixed hypothetical vector can be tested.
This choice will also be relevant for confidence regions for $\boldsymbol{p}$.
Alternatively, $\boldsymbol{C}$ may be chosen to be a contrast matrix, i.e., $\boldsymbol{C1}_d=\boldsymbol{0}_r$ where $ \boldsymbol{1}_d=(1, ..., 1)^\top \in \mathbb{R}^d$ and  $\boldsymbol{0}_r=(0, ..., 0)^\top \in \mathbb{R}^r$, if one is interested in certain differences of the Mann-Whitney effects.
Particular examples of usual hypothesis matrices in one- or two-way layouts will be given in subsequent section.}

\begin{rem}
In the case that $\vc=\vnull_r$, we assume w.l.o.g.\ that $\vC$ is a projection matrix because the unique projection matrix $\boldsymbol{P} =  \boldsymbol{C}^\top (\boldsymbol{C}\boldsymbol{C}^\top)^+ \boldsymbol{C}$ based on the Moore-Penrose pseudo-inverse formulates the same hypotheses since $\vP \vp = \vnull_d$ if and only if $\vC \vp = \vnull_r$. Also note that $\boldsymbol{P}$ is a contrast matrix if and only if $\boldsymbol{C}$ is one.

For hypotheses with $\vc\neq\vnull_r$, no convention for the hypothesis matrix is available. Although different pairs $(\vC,\vc)$ may describe the same hypothesis, there is in general no guarantee that an equivalent representation can be obtained through the projection matrix $\boldsymbol{P}$. Therefore, in this case the particular pair $(\vC,\vc)$ used in the analysis should be reported explicitly for reproducibility; see \cite{sattler2023} and \cite{sattler2025}.
\end{rem}

\begin{rem}
\label{rem:fay}
Using the same arguments as in \cite{fay18} combined with some straightforward algebra, it is easy to see that contrasts in the marginal effects can be written as follows:
$$ p_i - p_\ell =  E(H(X_{j1}) - H(X_{\ell 1})). $$
As a consequence, such contrasts also constitute causal estimands with an interpretation similar to the ``quantile difference causal effect'', denoted by $\phi$ in \cite{fay18}.

\textcolor{black}{Alternatively, one could consider the estimands $$ \theta_i = \frac1d \sum_{s=1}^d P(X_{s1} \preceq X_{i1}), \quad  i=1,\dots,d, $$
which might offer a more straightforward causal interpretation compared to $p_i$.
However, estimating $\theta_i$ would require a sufficiently large number of simultaneously observable pairs $(X_{sk}, X_{ik})$, i.e., $\lambda_{sk}=\lambda_{ik}=1$.
Hence, estimation of $\theta_i$ is more challenging (from the perspective of data availability) than the estimation of $p_i$; cf.\ Section~\ref{ssec:estimation} below.
Tests about $\theta_i$ would be close in spirit to multivariate extensions of the sign test.}
\end{rem}

\begin{example}
\label{ex:2sample}
Let us consider the simplest case, $d=2$, and let $X_1\sim F_1$ and $X_2\sim F_2$ be stochastically independent random variables. Then, for $Z \sim \tfrac12 (F_1 + F_2)$ stochastically independent of $(X_1,X_2)$,
    $$\boldsymbol{p} = \big(P(Z \preceq X_1), P(Z \preceq X_2)\big)^\top = \big(\tfrac14 + \tfrac12 P(X_2 \preceq X_1), \tfrac14 + \tfrac12 P(X_1 \preceq X_2)\big)^\top .$$
    Testing whether this vector equals $(\tfrac12, \tfrac12)^\top$ boils down to a test about the Mann-Whitney effect, i.e., $\tilde H_0: P(X_2 \preceq X_1) = \tfrac12$; see e.g.\ \cite{brumu00} in the case of two independent sample groups and \cite{dobler23} for a quasi-randomization test in the case of paired randomly right-censored time-to-event data.
    It is rather obvious how to derive one- and two-sided tests and confidence intervals for the Mann-Whitney effect in the context of the present example.
    These will result from Remark~\ref{rem:CI} below in a straightforward manner.

    Furthermore, 
suppose that $\vlam_k \stackrel{i.i.d.}{\sim} \eta \cdot \delta_{(0,1)^\top} + (1-\eta)\cdot \delta_{(1,0)^\top}, \eta \in (0,1) $,  where $\delta_{\vx}$ be the Dirac measure in $\vx \in \mathbb R^d$.
That is, exactly one value of each pair of outcome variables is observable, the other is not.
In this case, $\boldsymbol{p}$ is still identifiable.
Estimation of $\vp$ may essentially be based on the usual Mann-Whitney effect estimator in two independent sample groups with random sample sizes. 
\end{example}

Because the case of $d=2$ (without missingness) has already been thoroughly investigated in the literature, as we have seen in Example~\ref{ex:2sample} above, we will focus on the case $d>2$ (and of course $n \geq 2$) in the remainder of this paper.

\subsection{Estimation}
\label{ssec:estimation}

Throughout, we will try to suppress the notion of the sample size $n$ in the notation of all estimators for a lighter presentation. 
All estimators are presented with a `hat' above their symbol.
 The marginal Mann-Whitney effects can be estimated by plugging-in the empirical versions of $F_{i}$ and $H$, i.e.,  for $x \in \mathbb{R}$, 
\begin{align*}
\hat{F}_{i}(x)&:=\frac{1}{\lambda_{i.}}\sum_{k=1}^{n}\lambda_{ik}\hat{F}_{ik}(x)=\frac{1}{\lambda_{i.}}\sum_{k=1}^{n}\lambda_{ik}c(x-X_{ik}), 
\\
               \hat{H}(x)&:= \frac1d \sum_{i=1}^d \hat F_i(x) =\frac{1}{d}\sum_{i=1}^{d}\frac{1}{\lambda_{i.}}\sum_{k=1}^{n}\lambda_{ik}c(x-X_{ik}).
\end{align*}
Here, 
$c$ is the normalized version of the counting function, i.e., $c(u)=0, 1/2$, or $1$ according to whether $u<0, u=0$, or $u> 0$. 
Note that $\hat F_i(x)$ is well-defined only if there is at least one observation in each component.
Otherwise, we define $\hat p_i := \tfrac12$ because it seems impossible to favor the $i^{th}$ component over any other.
For $\lambda_{i\cdot}>0$, we estimate the
marginal Mann-Whitney effects $p_{i}$ by
\begin{align*}
\hat{p}_{i}&:=\int\hat{H}d\hat{F}_{i}=\frac{1}{\lambda_{i.}}\sum_{k=1}^{n}\lambda_{ik}\hat{H}(X_{ik})
=\frac{1}{\lambda_{i.}}\sum_{k=1}^{n}\frac{\lambda_{ik}}{d}\sum_{s=1}^{d}\frac{1}{\lambda_{s.}}\sum_{l=1}^{n}\lambda_{sl}c(X_{ik}-X_{sl}).
\end{align*}

In the following, some asymptotic theory  for the above estimators is developed. 
These results will be necessary for justifying the use of hypothesis tests derived from $\hat{\vp}=(\hat p_1, \dots, \hat p_d)^\top$.
To prepare this, we denote by $\stackrel p \rightarrow$ and $\stackrel d \rightarrow$ convergence in probability and in distribution, respectively.
For the asymptotics, we will rely on the following observability assumption:
\begin{assumption}\label{Assump1} 
$P(\lambda_{i1} = 1) > 0 $ \  for each \ $i=1, \dots, d$. 
\end{assumption}
{Under this assumption, $\vp$ is identifiable.}
As a consequence of Assumption~\ref{Assump1}, for some $\kappa_i \in (0,1]$, we have $\frac{\lambda_{i.}}{n} \stackrel p \rightarrow \kappa_{i} $ as $n\to \infty$ by the law of large numbers,  $i=1,..., d$.
 This ensures that, eventually, `sufficiently many' subjects will be observable in each component.

\begin{prop}\label{prop:bias}
 Under Assumptions~\ref{ass:MCAR} and~\ref{Assump1}, for all $i =1, \dots, d$, $\hat F_i,\hat H$, and $\hat p_i$ are asymptotically unbiased estimators for $F_i, H$, and $p_i$, respectively. 
 The biases of $\hat F_i$ and $\hat H$ decrease exponentially fast; the bias of $\hat p_i$ is of the order $o(n^{-1})$.
\end{prop}

\begin{rem}
Note that the biases in Proposition~\ref{prop:bias} are only due to the hypothetical possibility that $ \lambda_{i\cdot} = 0$, $i=1,\dots, n$.
The proof of Proposition~\ref{prop:bias} reveals that the bias of $\hat p_i$ could be reduced to an exponentially decreasing bias if $\hat p_i$ is re-defined in a way such that component comparisons within the same random vectors are avoided. This estimator would be more in the spirit of a $U$-statistic rather than a $V$-statistic. For simplicity, also in view of a computationally more efficient implementation, we continue with the originally definition of $\hat p_i$.    
\end{rem}

 \subsection{Asymptotics}
 \cite{rubarth22} derived a central limit theorem for $\sqrt{n}(\hat{\vp}-\vp)$ and tests about $\vp$ in the case of \emph{deterministic missingness indicators}.
They proved that, as $n \to\infty$, $\sqrt{n}(\hat{\vp}-\vp)$ asymptotically follows a multivariate normal distribution with expectation $\boldsymbol{0}_d$ and some covariance matrix $\vV_0$. This matrix reflects both the stochastic variability and the deterministic constraints of the estimator. The latter imply that $\vV_0$ is necessarily singular; below, we will generally assume that this singularity is minimal. In contrast to \cite{rubarth22}, where the missing indicators are assumed deterministic, in our setting they are explicitly incorporated into the model and, thus, substantially contribute to the asymptotic covariances. 

We arrive at our first main result, a central limit theorem for $\hat{\boldsymbol{p}}$.
Due to the present randomness of the $\lambda$'s, the asymptotic covariance matrix is different from $\vV_0$.

\begin{thm}
    \label{thm:clt}
    Under Assumptions~\ref{ass:MCAR} and~\ref{Assump1}, we have
    $$\sqrt{n}(\h \vp - \vp) \stackrel d \longrightarrow \mathcal{N}_d(\boldsymbol 0_d, \boldsymbol V)$$
    as $n \to \infty$, where $\vV \in \mathbb{R}^{d\times d}$ is a singular covariance matrix.
\end{thm}

\begin{rem}
\label{rem:CI}
    The derivation of elliptical confidence regions for $\boldsymbol{p}$ will be immediate from inverting a reasonable hypothesis test for $H_0(\boldsymbol{I}_d, \boldsymbol{p}_0),  \boldsymbol{p}_0 \in \mathbb{R}^d_1$.
    One- or two-sided confidence intervals for each $p_i \in [0,1]$ can be derived by choosing the $i^{th}$ unit vector, say $\boldsymbol{C} = \boldsymbol{e}_i \in \mathbb{R}^{1 \times d}$.
    \textcolor{black}{For instance, based on the asymptotic normality of each $\hat p_i$ and consistent variance estimators $\hat\sigma_i^2$, standard one-sided Wald-type confidence intervals are given by 
    $$ [0 \ , \ \hat p_i + n^{-1/2} \hat \sigma_i \Phi^{-1}(1-\alpha)] \quad \text{and} \quad [ \hat p_i + n^{-1/2} \hat \sigma_i \Phi^{-1}(\alpha) \ , \ 1], \quad i=1,\dots, d. $$
    Here, $\Phi$ denotes the cumulative distribution function of the standard normal distribution. Similarly, transformation-based variants of these intervals, e.g., using a complementary log-log-transformation, would force all boundaries into the interval $[0,1]$.
    Also, confidence intervals for single contrasts, $p_i-p_j$, $i\neq j$,  could be obtained by making the obvious adjustments.}

\end{rem}

    Because of the duality between hypothesis testing and confidence regions, we will mainly focus on testing in the following.
The asymptotic covariance matrix $\boldsymbol{V}$ is further specified in Appendix~\ref{app:covariance}.
It generally depends on unknown quantities and must be estimated for inferential applications of Theorem~\ref{thm:clt}. 
We propose a consistent estimator for it in Appendix~\ref{app:covariance}.
In any case, the rank $\mathrm{rk}(\boldsymbol{V})$ cannot be full, i.e., equal to $d$, since $\boldsymbol{1}^\top_d \hat{\boldsymbol{p}}\equiv \tfrac d2 $ is deterministic.
The subsequent Assumption~\ref{ass:test} specifies that minimal singularity is imposed to avoid asymptotic rank jumps of the asymptotic covariance matrix of $ \sqrt{n}(\h \vp - \vp) $, which would otherwise complicate inferential procedures.
\begin{assumption}
\label{ass:test}
    $\mathrm{rk}(\boldsymbol{V}) = d-1$.
\end{assumption}

Intuitively, Assumption~\ref{ass:test} holds whenever the data are sufficiently rich, in the sense that each coordinate~$i$ contains some randomness that cannot be fully explained by the remaining coordinates. Thus, under Assumption~\ref{ass:test},
$
\ker(\vV)=\{c\cdot\boldsymbol{1}_d: c\in\mathbb{R}\}.
$
 Due to this property, the image space of $\vV$ consists of contrasts, i.e.,
$
\im(\vV)
=
\{\vv\in\mathbb{R}^d:\vv^\top\boldsymbol{1}_d=0\}
=
\boldsymbol{1}_d^\perp.
$

Assumption~\ref{ass:test} is substantially less restrictive than corresponding assumptions about the underlying data. To illustrate this point, Appendix~\ref{app:rankexample} contains a four-dimensional example in which the covariance matrix of the underlying observable random vectors has rank~2, whereas the corresponding matrix $\vV$ has rank~3. We will later use a Wald-type test statistic which is based on a suitable consistent estimator $\hat{\vV}$ of $\vV$ with no asymptotic rank jumps: 

\begin{assumption}
\label{ass:V}
    $\mathrm{rk}(\hat{\vV}) \leq d-1$ and, as $n\to\infty$, $\hat{\vV} \stackrel p \to \vV$.
\end{assumption}
An important remark for the proofs in this paper is the fact that Assumption~\ref{ass:V} and the lower semi-continuity of the rank mapping imply that
$ \mathrm{rk}(\hat {\vV}) \stackrel p \to \mathrm{rk}(\vV) = d-1$ (by Assumption~\ref{ass:test}).
An estimator of $\boldsymbol{V}$ satisfying Assumption~\ref{ass:V} is proposed in Appendix~\ref{app:covariance}; see Theorem~\ref{thm:cov_consistency} therein.

The Wald-type test statistic is defined as
$$T_n(\boldsymbol{c}) = n (\vC \hat{\boldsymbol{p}} - \vc)^\top (\vC\hat{\boldsymbol V}\vC^\top)^{+} (\vC\hat{\boldsymbol{p}} - \vc).$$
Under $H_0: \vC \vp = \vc$, the statistic can be written as 
$T_n \stackrel{H_0}= n (\hat{\boldsymbol{p}} - \boldsymbol{ p})^\top \vC^\top (\vC\hat{\boldsymbol V}\vC^\top)^{+} \vC(\hat{\boldsymbol{p}}- \boldsymbol{p}).$
We will see that, under the mentioned assumptions, $T_n$ is asymptotically $\chi^2$-distributed under $H_0$ and it will diverge under $H_1$. 
The asymptotic degrees of freedom equal $f(\vC) := \mathrm{rk}(\vC \vV) = \mathrm{dim}(\vC \boldsymbol{1}_d^\perp) \leq d-1$.
Here it is important to note that $f(\vC)$ does not depend on any unknown parameters.
Denoting by $\chi^2_f$ the chi-squared distribution with $f\in \mathbb N$ degrees of freedom, we arrive at the following large-sample properties of $T_n$.
In order to exclude trivial cases such as $\vC = (1,\dots, 1)$, we assume throughout that $f(\vC) \geq 1$.

\begin{thm}
\label{thm:T}
Under Assumptions~\ref{ass:MCAR}--\ref{ass:V} and as $n\to\infty$,
\begin{enumerate}
        \item $T_n \stackrel d\to \chi^2_{f(\vC)}$ under $H_0$;
        \item $T_n \stackrel p\to \infty$ under $H_1$.
\end{enumerate}  
\end{thm}
This result enables the development of inferential procedures based on the limiting distributions.

\section{Quasi-randomization tests}
\label{sec:rand_inference}

Typically, the approximation quality of a test statistic's null distribution by the estimated normal limit distribution is improvable by means of resampling methods.
In this paper, we pursue the idea of a quasi-randomization-type resampling method \citep{dobler23}.
It is motivated by a special case under the null hypothesis: if the data components are exchangeably distributed, this would in particular imply that the marginal distributions are the same. Hence, all marginal Mann-Whitney effects would be equal and in some cases the null hypothesis is implied for the quasi-randomization-based Mann-Whitney effects vector, e.g., when $\vC$ is a contrast matrix and $\vc = \boldsymbol 0$.
This feature paves the way for inference procedures with finite sample guarantees under component exchangeability.
At the same time, it remains to ensure that asymptotic correctness is guaranteed if the special case of exchangeability does not hold.

To prepare the quasi-randomization approach, we state the asymptotic linearity of the estimator, 
$$\sqrt{n}(\hat {\boldsymbol{p}} - \boldsymbol{p}) = \frac1{\sqrt{n}} \sum_{k=1}^n ( \dot\psi_{\mathbb{P}} (\boldsymbol{X}_k, \boldsymbol{\lambda}_k) - E(\dot\psi_{\mathbb{P}} (\boldsymbol{X}, \boldsymbol{\lambda}))) + \lan_p(1).$$
Here, $\dot\psi_{\mathbb{P}}$ are the so-called \emph{influence functions}; it turns out that the expectation equals zero. We refer to Appendix~\ref{app:proofs} for details.

Let us now introduce the resampling technique of our choice to approximate the null distribution of the above-introduced Wald-type test statistic $T_n$.
Instead of the classical bootstrap, we aim to reproduce the missingness mechanisms in a better way by means of a certain permutation technique.
Our general approach is to randomly permute the components of each observed data point.
This kind of \emph{within-individual} permutation satisfies the general quasi-randomization set-up of \cite{dobler23} and thus makes his theory directly applicable.

For each data point, $(\vX_{k}, \vlam_k)$, we introduce a random permutation of its component through 
$$\Pi_{k}: \mathbb R^{2d} \to \mathbb R^{2d}, \ (x_1, \dots, x_d, \ell_1, \dots, \ell_d) \longmapsto (x_{\pi(1)}, \dots, x_{\pi(d)}, \ell_{\pi(1)}, \dots, \ell_{\pi(d)}) .$$
Here, $\pi = \pi^{(k)}$ are independent random elements of the symmetric group $\mathcal{S}_d$ of degree $d$, i.e., random permutations of the numbers $1,\dots,d$.
Note that the missingness indicators $(\lambda)$ are permuted in the same way as the corresponding outcome values $(X)$.

In order to avoid confusion with classical permutation tests based on random permutations between sample groups from multiple populations, as well as with Fisherian design-based randomization tests relying on a physical randomization mechanism in the experiment \citep{fisher1935}, we will henceforth use the term \emph{quasi-randomization} to address the present approach. 
On this occasion, we would also like to point to the ongoing discussion of the differing terminology \citep{zhang23,Goemann2021,hemerik24}. For the connection between quasi-randomization tests (or ``randomization-based inference'') and randomization in experimental designs, see \cite{rosenberger15}, in particular Section~6.4 therein, and The 41${}^{\textnormal{st}}$ Fisher Memorial Lecture \citep{rosenberger26}.

The quasi-randomized data points are thus given as
$(\vX_{k}^\pi, \boldsymbol \lambda_{k}^\pi) := \Pi_{k}(\vX_{k}, \vlam_k) \stackrel{i.i.d.}\sim \tilde{\mathbb P}$, $k=1,\dots,n$.
To be more precise, the quasi-randomized missingness indicators $\boldsymbol \lambda_{k}^\pi$ consist of
$$ 
 \lambda_{jk}^\pi:=  \begin{cases}
	 1, \quad \text{if $X_{\pi(j)k}$ is observed}	\\
	 	 0, \quad \text{if $X_{\pi(j)k}$ is non-observed}
 \end{cases}
 j=1, ...,d, \ k=1, ...,n.
$$ 
Based on the quasi-randomized random variables, we introduce the quasi-randomization versions of $\hat F_{j}$, $\hat H$, and $\hat p_{j}$ in the obvious way by re-computing these estimators based on the quasi-randomized variables and denote them $\hat F_{j}^\pi$, $\hat H^\pi$, and $\hat p_{j}^\pi$, respectively.
We denote the corresponding vectors as $\hat {\vF}^\pi:=(\hat F_{1}^\pi, \dots, \hat F_{d}^\pi)^\top$ and $\hat {\vp}^\pi := (\hat {p}_{1}^\pi, \dots, \hat {p}_{d}^\pi)^\top$, respectively.
We will generally try to denote quasi-randomized estimators with a hat symbol combined with the superscript $\pi$ while their limits are denoted with a tilde symbol. Conditional expectations will be denoted with a superscript $\pi$, without the hat symbol.

\textcolor{black}{The conditional expectations of $\hat F_{j}^\pi$ and $\hat p_{j}^\pi$ given the $\sigma$-algebra $\mathcal{A}_{X,\lambda} := \sigma(\vX_{k}, \boldsymbol \lambda_{k}: k=1,\dots, n)$ do not depend on the coordinate $j$. That is, $ F_{\cdot}^\pi := E(\hat F_{j}^\pi \ | \ \mathcal{A}_{X,\lambda})$ and $ p_{\cdot}^\pi := E(\hat p_{j}^\pi \ | \ \mathcal{A}_{X,\lambda})$.
 Furthermore,  $p_{\cdot}^\pi = E(\tfrac1d\sum_{j=1}^d \hat p_{j}^\pi \ | \ \mathcal{A}_{X,\lambda}) = E( \int \hat H^\pi d \hat H^\pi \ | \ \mathcal{A}_{X,\lambda}) = E( \tfrac12 \ | \ \mathcal{A}_{X,\lambda}) = \tfrac12$.}

The following theorem establishes a conditional central limit theorem for the quasi-randomized Mann-Whitney effect estimator.
The conditional weak convergence in outer probability therein is to be understood as the convergence of conditional distributions in probability on a metric space of probability measures;
for the concept of outer probabilities, we refer to \cite{vdVW23}.

\begin{thm}
\label{thm:main_rand}
{Under Assumptions~\ref{ass:MCAR}--\ref{ass:test} and}
conditionally on $\mathcal{A}_{X,\lambda}$,
the following convergence in distribution holds in outer probability as $n \to\infty$:
$$ \sqrt{n} (\hat {\textbf{p}}^\pi - \tfrac12 \cdot \boldsymbol{1}_d) \stackrel{d}\to \tilde{\vZ} \sim \mathcal{N}_d(\boldsymbol 0_d, \tilde{\boldsymbol V}) ;$$
the limiting variance-covariance matrix is given by
$$\tilde {\vV} = \Bigg(\begin{smallmatrix}
\tilde \sigma^2 & \tilde \rho & 
\dots & \tilde \rho \\
 \tilde \rho & \tilde \sigma^2 & 
  \dots & \tilde \rho \\
 \vdots & \vdots & 
 \ddots & \vdots \\
 \tilde \rho & \tilde \rho & 
  \dots &  \tilde \sigma^2 
\end{smallmatrix}\Bigg),$$ 
it holds that $\tilde \sigma^2 > 0$ and $\tilde \rho = -\tilde \sigma^2/(d-1)$.
Thus, $\mathrm{rk}(\tilde {\boldsymbol{\vV}})=d-1$.
\end{thm}

For constructing a quasi-randomization Wald-type test for $H_0$ versus $H_1$ we will now use the quasi-randomization version of our test statistic given by
$$T_n^\pi := n (\hat{\boldsymbol{p}}^\pi - \tfrac12 \cdot \boldsymbol{1}_d)^\top \vC^\top (\vC\hat{\boldsymbol V}^\pi\vC^\top)^{+} \vC(\hat{\boldsymbol{p}}^\pi- \tfrac12 \cdot \boldsymbol{1}_d),$$
where $\hat{\boldsymbol V}^\pi$ is the quasi-randomization version of $\hat{\boldsymbol V}$, i.e., $\hat{\boldsymbol V}$ calculated based on the quasi-randomized sample.
{If $\vC$ is a contrast matrix, $T_n^\pi$ simplifies to $T_n^\pi = n (\hat{\boldsymbol{p}}^\pi)^\top \vC^\top (\vC\hat{\boldsymbol V}^\pi\vC^\top)^{+} \vC \hat{\boldsymbol{p}}^\pi$.}

\textcolor{black}{Following the arguments of \cite{dobler23}, we define $\hat{\vV}$ to be the estimated empirical second moments of the influence function $\dot\psi_{\mathbb P}(\vX_k, \vlam_k)$; cf.\ Appendix~\ref{app:covariance} below. Exploiting the continuity of the involved functionals, it is straightforward to argue the consistency of $\hat{\vV}$. 
Similarly, $\hat{\vV}^\pi$ is a continuous functional of the randomization empirical process \citep{dobler23}.
Thus, under Assumptions~\ref{ass:test} and~\ref{ass:V}, we obtain that $\hat{\vV}^\pi \stackrel p \to \tilde{\vV}$ and $\mathrm{rk}(\hat{\vV}^\pi) \stackrel p \to d-1$ as $n\to\infty$.}
The following Corollary~\ref{cor:T_pi} summarizes the relevant large-sample properties of $T_n^\pi$.
\begin{cor}
\label{cor:T_pi}
Under Assumptions~\ref{ass:MCAR}--\ref{ass:V} and as $n\to\infty$, the following convergences hold conditionally on $\mathcal{A}_{X,\lambda}$ in outer probability under $H_0$ and $H_1$: $T_n^\pi \stackrel d\to \chi^2_{f(\boldsymbol{C})}$.
\end{cor}
Based on these results, we immediately derive the asymptotic type-I error control and the consistency of the deduced hypothesis test $ \Psi_n^\pi := \ind_{\{T_n > c_{1-\alpha,n}^\pi\}} $.
Here, we define $\Psi_n^\pi := 0$, whenever there is no observation in at least one component, and $T_n^\pi := \infty$ whenever this happens after a quasi-randomization.
These choices potentially make the test more conservative in cases where it is not sensible to draw a conclusion.

\begin{cor}
\label{cor:test1} 
As $n \to \infty$,
under Assumptions~\ref{ass:MCAR}--\ref{ass:V}, $E(\Psi^\pi_n) \to \alpha\cdot \ind_{H_0} + \ind_{H_1}$ for each significance level $\alpha\in (0,1)$.
In addition, $E(\Psi^\pi_n) \leq \alpha$  for every finite $n\geq 2$  if the data points are permutation-invariant, $ (\boldsymbol{X}_k, \boldsymbol{\lambda}_k) \stackrel d= (\boldsymbol{X}_k^\pi, \boldsymbol{\lambda}_k^\pi)  $ for each fixed permutation $\pi \in \mathcal{S}_d$.
\end{cor}

\begin{rem}\label{rem:CIrandom}
The quasi-randomization result also provides a natural way to obtain confidence
regions. Namely, the inversion of the Wald-type statistic $T_n$ can be based on the
conditional quantiles of its quasi-randomization version $T_n^\pi$ rather than on the quantiles
of the limiting $\chi^2$-distribution.
An analogous construction applies to one-sided intervals: the standard normal
quantiles are replaced by the corresponding conditional quantiles of
$
\sqrt{n}(\hat p_i^\pi-\tfrac12)
/ \hat{\sigma}_i^\pi$,
where $\hat{\sigma}_i^\pi$ is the quasi-randomization version of the standard deviation estimator $\hat{\sigma}_i$.

\end{rem}

\subsection{Factorial designs: one-way layout, or: repeated measures}
\label{ssec:one-way}
We first consider the simple but frequently encountered case of a one-way layout, where the components correspond to $d$ levels of a single factor.
The factor levels may represent, for instance, different time points, treatment conditions, dose levels, measurement methods, raters, or stimulus categories. The corresponding hypothesis typically refers to the absence of a systematic effect of this factor, i.e., the marginal effects are equal across all $d$ levels. In this setting several hypotheses may be of interest:

\begin{example}
{For testing the equality of all components of the vector $\vp$, one may use the following unique projection hypothesis matrix: $\vC_1 := \boldsymbol{P}_d = \boldsymbol{I}_d - \tfrac1d \boldsymbol{J}_d$, the so-called \textit{centering matrix}.
    It uses the identity matrix $\boldsymbol{I}_d \in \mathbb{R}^{d\times d}$ and
    the matrix of ones, $\boldsymbol{J}_d = \boldsymbol{1}_d \boldsymbol{1}_d^\top \in \mathbb{R}^{d\times d}$.}

    {If the measurements were taken from two exams per week, say, Mondays and Fridays, for each of $d/2 \in \mathbb{N}$ consecutive weeks, one could test the within-week and across-week equalities by using the projection matrix resulting from the contrast matrices
    \begin{align*}
        \vC_2 :=\Big(\begin{smallmatrix}
            1 & -1 & 0 & 0 & 0 & \dots & 0 \\
            0 & 0 & 1 & -1 & 0 & \dots & 0 \\
            \vdots & \vdots & \vdots & \vdots & \ddots & \vdots 
        \end{smallmatrix}\Big)
        \quad \text{and} \quad 
        \vC_3 :=\Bigg(\begin{smallmatrix}
            1 & 0 & -1 & 0 & 0 & 0 & 0 &  \dots   & 0 \\
            0 & 1 & 0 & -1 & 0 & 0 & 0 & \dots & 0 \\
            0 & 0 & 1 & 0 & -1 & 0 & 0 &  \dots & 0 \\
            0 & 0 & 0 & 1 & 0 & -1 & 0 & \dots & 0 \\
            \vdots & \vdots & \vdots & \vdots & \ddots & \ddots & \ddots & \ddots & \vdots 
        \end{smallmatrix} \Bigg),
    \end{align*}
    respectively. 
    If one wishes to test early weeks versus late weeks, one may use the projection matrix resulting from the contrast $\vC_4 :=(1 \ 1 \ \dots \ 1 \ -1 \ -1 \ \dots \ -1)$. Such problems already touch upon the two-way layouts which will be treated in more detail in the subsequent subsection.}
\end{example}

\textcolor{black}{Because all of the above-mentioned hypotheses matrices are contrast matrices, the asymptotic degrees of freedom of $T_n$ are $f(\vC) = \min(r,d-1)$, where $r$ is the number of rows of $\vC$.}

\subsection{Factorial designs: two-way layout and alternative quasi-randomization}
\label{ssec:two-way}

Similar results hold for higher-way layouts.
We exercise our approach in the context of a two-way layout in which we test for a main effect in the first factor.
To be precise, we represent the data as 
$$ (\boldsymbol{X}_k, \boldsymbol{\lambda}_k) = (X_{11k}, X_{12k}, \dots, X_{1d_Bk}, X_{21k}, \dots, X_{d_Ad_Bk}, \lambda_{11k}, \lambda_{12k}, \dots, \lambda_{1d_Bk}, \lambda_{21k}, \dots, \lambda_{d_Ad_Bk}),
$$
that is, $k$ independent $2d_A d_B$-dimensional data points, where factors $A$ and $B$ consist of $d_A$ and $d_B$ levels, respectively.
We denote by $p_{ij}$ the unweighted marginal Mann-Whitney effect of the factor combination $(i,j)$ when compared to the average of all factor combinations:
$$ p_{ij} = \int H d F_{ij} \quad \text{where}  \quad H = \frac1d \sum_{i=1}^{d_A} \sum_{j=1}^{d_B} F_{ij}; $$
here and below we use the obvious extension of the notation used in the one-way layout.

The null hypotheses for a main effect in factor $A$, factor $B$, 
and an interaction effect are 
\begin{itemize}
    \item $H_{0A}(\boldsymbol{C}_{A}): \{ \boldsymbol{C}_{A} \boldsymbol{p} = \boldsymbol{0}_{d_Ad_B} \}$  \ where \  $\boldsymbol{C}_{A} = \boldsymbol{P}_{d_A} \otimes \boldsymbol{1}_{d_B}^\top  \in \mathbb{R}^{d_A \times d_A d_B} $;
    \item $H_{0B}(\boldsymbol{C}_{B}): \{ \boldsymbol{C}_{B} \boldsymbol{p} = \boldsymbol{0}_{d_Ad_B} \}$ \ where \ 
    $\boldsymbol{C}_{B} = \boldsymbol{1}_{d_A}^\top \otimes
        \boldsymbol{P}_{d_B}  \in \mathbb{R}^{d_B \times d_A d_B} $;
    \item $H_{0AB}(\boldsymbol{C}_{AB}): \{\boldsymbol{C}_{AB} \boldsymbol{p} = \boldsymbol{0}_{d_Ad_B} \}$ \ where \ $\boldsymbol{C}_{AB} = \boldsymbol{P}_{d_A} \otimes \boldsymbol{P}_{d_B}  \in \mathbb{R}^{d_A d_B \times d_A d_B} $.
\end{itemize}
Theorem~\ref{thm:main_rand} is readily extended to the two-way layout.
Here, Assumption~\ref{ass:test} would translate to $\mathrm{rk}(\boldsymbol{V}) = d_A d_B-1$.

The quasi-randomization scheme of Section~\ref{ssec:one-way} can be applied in exactly the same way for inference about any of the three null hypotheses in the two-way layout.
This simply follows from splitting up the indices of a one-way layout to convert it into a two-way one.
However, quasi-randomizing all $d_Ad_B$ entries would result in quasi-randomized marginal Mann-Whitney effects which are all equal, i.e., $ p_{\cdot}^\pi = \tfrac12$, and the exchangeability of the distribution of $\hat\vp^\pi$.
In view of testing for a main effect, say $A$, this seems unnecessarily restrictive.
Instead, if a quasi-randomization scheme would only establish the equalities which are reflected under $H_{0A}$, 
\textcolor{black}{i.e., $p_{1j}^\pi = p_{2j}^\pi = \dots = p_{d_A j}^\pi$, $j=1,\dots, d_B$,}
then the finite sample control of the type-I error probability of the resulting quasi-randomization hypothesis test will be much more likely to hold, i.e., just under the interchangeability according to factor $A$.
\textcolor{black}{Due to the symmetry of the factors $A$ and $B$, we will only focus on the factor $A$ henceforth.}

The invariance idea underlying the following approaches is closely related to that used in tests for interchangeability \citep{jerdack87,jerdack90}, i.e., only certain fixed permutations within the original random vectors retain the same joint distribution, but not necessarily all.
Such interchangeability is motivated by studies following block designs.
As a consequence, the quasi-randomization tests developed in the following control the type-I error probability if only the corresponding interchangeability holds, which is less restrictive than full exchangeability.

\subsubsection{Testing main effects in two-way layouts}\label{twoway}

To make the adjusted quasi-randomization scheme precise, we propose the following quasi-randomization scheme:
let $\pi = \pi^{(k)}: \Omega \to \mathcal{S}_{d_A}$ be independent random permutations of the numbers $1, \dots, d_A$, \ $k=1,\dots,n$.
The quasi-randomized sample is given by
\begin{align}
\begin{split}
\label{eq:permA}
    (\boldsymbol{X}_{k}^\pi, \vlam_k^\pi) :=  (&X_{\pi(1) 1 k}, X_{\pi(1) 2 k}, \dots, X_{\pi(1) d_B k}, X_{\pi(2) 1 k}, \dots, X_{\pi(d_A) d_B k}, \\
    & \lambda_{\pi(1) 1 k}, \ \lambda_{\pi(1) 2 k}, \ \dots, \lambda_{\pi(1) d_B k}, \ \lambda_{\pi(2) 1 k}, \ \dots, \lambda_{\pi(d_A) d_B k}), \ \qquad k=1,\dots, n. 
\end{split}
\end{align} 
Based on these quasi-randomized samples, for each $j=1, \dots, d_B$, the re-computed estimated marginal Mann-Whitney effects, i.e., $\hat p^\pi_{1j}, \dots, \hat p^\pi_{d_Aj}$, will converge to the same value in probability due to $E(\hat p^\pi_{ij} \ | \ \mathcal A_{X,\lambda} ) =: p_{\cdot j}^\pi$ being independent of $i \in \{1,\dots,d_A\}$.
Thus, the proposed quasi-randomization establishes $H_{0A}$.
We will use the notation ${\boldsymbol{p}}^\pi := \boldsymbol{1}_{d_A} \otimes (p_{\cdot 1}^\pi, p_{\cdot 2}^\pi, \dots, p_{\cdot d_B}^\pi)^\top$.

Furthermore, \textcolor{black}{using a variant of Theorem~\ref{thm:main_rand},} it can be shown that the conditional asymptotic distribution of 
$ \sqrt{n} (\hat{\boldsymbol{p}}^\pi - \textcolor{black}{{\boldsymbol{p}}^\pi})$ \textcolor{black}{converges in probability to a multivariate normal distribution} with zero mean vector and variance-covariance matrix $\tilde {\boldsymbol{V}}$.
Here, $\tilde {\boldsymbol{V}}$ is given by the following block compound symmetry structure:
$$ \tilde {\boldsymbol{V}} = \boldsymbol{I}_{d_A} \otimes \tilde{\boldsymbol{\Sigma}} + (\boldsymbol{J}_{d_A} - \boldsymbol{I}_{d_A}) \otimes \tilde{\boldsymbol{\Upsilon}}  =  \begin{pmatrix}
                    \tilde{\boldsymbol{\Sigma}} & \tilde{\boldsymbol{\Upsilon}} & \dots & \tilde{\boldsymbol{\Upsilon}} \\
                    \tilde{\boldsymbol{\Upsilon}} & \tilde{\boldsymbol{\Sigma}} & \dots & \tilde{\boldsymbol{\Upsilon}} \\
                    \vdots & \vdots & \ddots & \vdots \\
                    \tilde{\boldsymbol{\Upsilon}} & \tilde{\boldsymbol{\Upsilon}} & \dots & \tilde{\boldsymbol{\Sigma}}
                   \end{pmatrix} \in \mathbb{R}^{(d_Ad_B) \times (d_A d_B)}, $$
for some symmetric matrices  $\tilde{\boldsymbol{\Sigma}}, \tilde{\boldsymbol{\Upsilon}} \in \mathbb{R}^{d_B \times d_B}$.

Under the above-mentioned variant of Assumption~\ref{ass:test}, asymptotic results similar to those in Theorem~\ref{thm:T} and Corollary~\ref{cor:test1} can be shown. This yields the following Corollary~\ref{cor:A} for the specialized quasi-randomization scheme introduced in this subsection. Note that since $\vp^\pi \in \ker(\vC_A)$ almost surely, we may state the result in terms of the reduced representation of $T_n^\pi$.

\begin{cor}
\label{cor:A}
Under Assumptions~\ref{ass:MCAR}--\ref{ass:V} and as $n\to\infty$, the following convergence holds conditionally on $\mathcal{A}_{X,\lambda}$ in outer probability under $H_0$ and $H_1$: $$T_n^\pi =  n (\hat{\boldsymbol{p}}^\pi)^\top \vC_A^\top (\vC_A\hat{\boldsymbol V}^\pi\vC_A^\top)^{+} \vC_A \hat{\boldsymbol{p}}^\pi \stackrel d\to \chi^2_{d_A-1}.$$
Consequently, for the resulting quasi-randomization test, say $ \Psi_n^{\pi\min}$, we have $E(\Psi^{\pi\min}_n) \to \alpha\cdot \ind_{H_0} + \ind_{H_1}$ for each significance level $\alpha\in (0,1)$.
In addition, $E(\Psi_n^{\pi\min}) \leq \alpha$  for every finite $n\geq 2$  if the distribution of $(\boldsymbol{X}_1,\vlam_1)$ is invariant under the vector-wise permutation of the first factor, i.e., $\vX_1 \stackrel d = \vX^\pi_1 $ and $\vlam_1 \stackrel d= \vlam_1^\pi$ in \eqref{eq:permA} for all \textit{fixed} permutations $\pi \in \mathcal{S}_{d_A}$.
\end{cor}
The asymptotic degrees of freedom are derived in Appendix~\ref{ssec:rank_main}. An analogous statement is obtained when the roles of factors (A) and (B) are interchanged.

\subsubsection{Testing interaction effects in two-way layouts}

An adjusted quasi-randomization procedure for testing $H_{0AB}$, the null hypothesis of no interaction, works as follows:
let $\pi_A = \pi_A^{(k)}: \Omega \to \mathcal{S}_{d_A}$ and $\pi_B = \pi_B^{(k)}: \Omega \to \mathcal{S}_{d_B}$ be independent random permutations of the numbers $1, \dots, d_A$ and $1, \dots, d_B$, respectively, $k=1,\dots,n$.
Then the quasi-randomized sample is given by
\begin{align}
\begin{split}
\label{eq:permAB}
    (\boldsymbol{X}_{k}^\pi, \vlam_k^\pi):=  (&X_{\pi_A(1) \pi_B(1) k}, X_{\pi_A(1) \pi_B(2) k}, \dots, X_{\pi_A(1) \pi_B(d_B) k}, X_{\pi_A(2) \pi_B(1) k}, \dots, X_{\pi_A(d_A) \pi_B(d_B) k}, \\
     \lambda_{\pi_A(1) \pi_B(1) k}, &  \lambda_{\pi_A(1) \pi_B(2) k}, \dots, \lambda_{\pi_A(1) \pi_B(d_B) k}, \lambda_{\pi_A(2) \pi_B(1) k}, \dots, \lambda_{\pi_A(d_A) \pi_B(d_B) k}), \ k=1,\dots, n. 
\end{split}
\end{align} 
Based on these quasi-randomized samples, the re-computed estimated marginal Mann-Whitney effects $\hat p^\pi_{ij}$ all converge to the same limit in probability: $p^\pi := E(\hat p^\pi_{ij} \ | \ \mathcal{A}_{X,\lambda}) = \tfrac12$, $i=1,\dots, d_A$, \  $j=1, \dots, d_B$.
Thus, the proposed quasi-randomization scheme establishes $H_{0AB}$.
Furthermore, it creates a situation of \textit{block-wise {interchangeability}}: all factor $A$ entries are permuted simultaneously and, independently thereof, also all factor $B$ entries are permuted simultaneously. 
The algebraic group induced by this quasi-randomization procedure has a cardinality of $d_A! d_B!$ which is much smaller than the `quasi-randomize everything' approach with an algebraic group of cardinality $(d_Ad_B)!$.
As a consequence, the resulting hypothesis test for $H_{0AB}$ based on the reduced quasi-randomization will be finitely control the type-I error probability under the much larger subpopulation where only block-wise {interchangeability} holds, as opposed to the very small subpopulation of complete exchangeability.
In the subsequent section, we will also investigate in a simulation study whether the reduced quasi-randomization scheme bears an advantage in situations under $H_{0AB}$ where neither exchangeability nor interchangeability  hold.

We will only consider the reduced quasi-randomization scheme henceforth.
Conditionally on the data, $\hat {\boldsymbol{p}}^\pi$ has an asymptotic block compound symmetry covariance matrix:
$$ \tilde {\boldsymbol{V}} = \boldsymbol{I}_{d_A} \otimes \tilde{\boldsymbol{\Sigma}} + (\boldsymbol{J}_{d_A} - \boldsymbol{I}_{d_A}) \otimes \tilde{\boldsymbol{\Upsilon}} =  \begin{pmatrix}
                    \tilde{\boldsymbol{\Sigma}} & \tilde{\boldsymbol{\Upsilon}} & \dots & \tilde{\boldsymbol{\Upsilon}} \\
                    \tilde{\boldsymbol{\Upsilon}} & \tilde{\boldsymbol{\Sigma}} & \dots & \tilde{\boldsymbol{\Upsilon}} \\
                    \vdots & \vdots & \ddots & \vdots \\
                    \tilde{\boldsymbol{\Upsilon}} & \tilde{\boldsymbol{\Upsilon}} & \dots & \tilde{\boldsymbol{\Sigma}}
                   \end{pmatrix} \in \mathbb{R}^{(d_Ad_B) \times (d_A d_B)}. $$
Here, each matrix block also exhibits a compound symmetry structure:
 $$\tilde{\boldsymbol{\Sigma}} = \begin{pmatrix}
                    \tilde \sigma^2 & \tilde \rho & \dots & \tilde\rho \\
                    \tilde\rho & \tilde \sigma^2 & \dots & \tilde\rho \\
                    \vdots & \vdots & \ddots & \vdots \\
                    \tilde\rho & \tilde\rho & \dots & \tilde \sigma^2
                   \end{pmatrix} \in \mathbb{R}^{d_B \times d_B} \quad \text{and} \quad \tilde{\boldsymbol{\Upsilon}} = \begin{pmatrix}
                    \tilde\eta & \tilde\nu & \dots & \tilde\nu \\
                    \tilde\nu & \tilde\eta & \dots & \tilde\nu \\
                    \vdots & \vdots & \ddots & \vdots \\
                    \tilde\nu & \tilde\nu & \dots & \tilde\eta
                   \end{pmatrix} \in \mathbb{R}^{d_B \times d_B}.$$
Furthermore, due to $\tilde {\boldsymbol{V}} \boldsymbol{1}_{d_Ad_B} = \boldsymbol{0}_{d_Ad_B}$, it holds that
$\tilde \sigma^2 + (d_A-1) \tilde\eta + (d_B-1) \tilde\rho  + (d_A-1)(d_B-1) \tilde\nu = 0.$
The limiting distribution of $T_n^\pi$ in the two-way layout version of
Theorem~\ref{thm:T} is \(\chi^2_{(d_A-1)(d_B-1)}\) under the corresponding
variant of Assumption~\ref{ass:test}.
The following result refers to the specialized quasi-randomization scheme introduced in the present subsection.
\begin{cor}
Under Assumptions~\ref{ass:MCAR}--\ref{ass:V} and as $n\to\infty$, the following convergence holds conditionally on $\mathcal{A}_{X,\lambda}$ in outer probability under $H_0$ and $H_1$: 
$$T_n^\pi = n (\hat{\boldsymbol{p}}^\pi)^\top \vC_{AB}^\top (\vC_{AB}\hat{\boldsymbol V}^\pi\vC_{AB}^\top)^{+} \vC_{AB} \hat{\boldsymbol{p}}^\pi \stackrel d\to \chi^2_{(d_A-1)(d_B-1)}.$$
Consequently, for the resulting quasi-randomization test, say $\Psi_n^{\pi\min}$, we have $E(\Psi_n^{\pi\min}) \to \alpha\cdot \ind_{H_0} + \ind_{H_1}$ for each significance level $\alpha\in (0,1)$.
In addition, $E(\Psi_n^{\pi\min}) \leq \alpha$  for every finite $n\geq 2$  if the distribution of  $(\boldsymbol{X}_1, \boldsymbol{\lambda}_1)$ is invariant under the vector-wise permutation of the first factor and also of the second factor, i.e., $\vX_1 \stackrel d = \vX^\pi_1 $ and $\vlam_1 \stackrel d= \vlam_1^\pi$ in \eqref{eq:permAB} for all \textit{fixed} permutations $\pi_A \in \mathcal{S}_{d_A}$ and $\pi_B \in \mathcal{S}_{d_B}$.
\end{cor}

The derivation of the asymptotic degrees of freedom is given in Appendix~\ref{ssec:rank_interaction}.
Similar approaches are applicable for testing the various obvious hypotheses in higher-way layouts. We do not further exercise these ideas here for the sake of brevity.

\section{Simulation study}\label{simulation}
To compare the finite-sample performance of the proposed tests in terms of their sizes and power, we conducted a simulation study. 
We consider the asymptotic test $\Psi^\textnormal{asy}_n$, the quasi-randomization test $\Psi^\pi_n$, and also a naive bootstrap test $\Psi^\textnormal{boot}_n$ 
which respectively reject the null hypothesis if the Wald-type test statistic $T_n$ exceeds the quantiles of the $\chi_{f(\vC)}^2$-distribution, the $(1-\alpha)$-quantile $c_{1-\alpha}^\pi$ of the conditional distribution of $T_n^\pi$ given the data, and the $(1-\alpha)$-quantile $c_{1-\alpha}^{\textnormal{boot}}$ of the conditional distribution of a bootstrap-variant $T_n^\textnormal{boot}$ given the data.
Here, similar to the quasi-randomization approach, bootstrap samples are generated by repeatedly and randomly drawing from all components of each data point with replacement; for each drawing, the same drawn indices are used for the observations ($x$) and the corresponding missingness indicators ($\lambda$). 
For the two-factor setting, we additionally simulated the reduced variant $\Psi^{\pi\min}_n$ from Section~\ref{twoway} which is based on the minimized quasi-randomization groups.
Thus, the concrete choice of quasi-randomization for this test depends on the considered null hypothesis.
We chose the nominal significance level $\alpha=0.05$ for all tests.

 For each simulation setting described in the subsequent subsections, the empirical rejection rates are based on \(5{,}000\) simulation runs; the critical values of the resampling-based tests are derived from \(2{,}000\) bootstrap or quasi-randomization repetitions, respectively. 

\subsection{Type-I error probability simulation}
\label{sec:type-I}
First, realizations of $\vZ_1, \dots, \vZ_n$ are generated with the help of  the R package \texttt{copula} \citep{hofert2026copula};
note that equal marginal cumulative distribution functions do not change the value of the Mann-Whitney effect.
Thus, we chose Archimedean Clayton and Gumbel copulas with Kendall's rank correlation coefficients $\tau\in \{0.1,0.2\}$. For both types of copulas, the corresponding copula parameter is obtained from the respective value of $\tau$ using the function \textit{iTau} from the same R package.
Consequently, we obtained data with lower tail dependence (Clayton) and upper tail dependence (Gumbel). 
To obtain data with not non-exchangeable and non-interchangeable components, therewith we calculate the $i^{th}$ component of $\vX_k$ as 
$$X_{ki}=(1+ 0.1(i-1))((\vZ_{k})_i-1/2)$$
for $k=1,...,n$ and $ i=1,...,d$, which does not affect the Mann-Whitney effects in the present case.

Two different settings are studied: a one-way layout with \(d_A = 11\) levels, resulting in  11-dimensional data---this setting is motivated from the real data example in Section~\ref{application} below---and a two-way layout with $d_A = 3$ levels for factor $A$ and $d_B = 2$ levels for factor $B$, resulting in $d = d_A d_B = 6$-dimensional data.
We consider tests for $H_{0,\textnormal{equal}}: p_1 = \dots = p_{11}$, the null hypothesis of equality, and the null hypotheses $H_{0,\textnormal{equal}}: p_{11} = \dots = p_{d_A d_B}$, $H_{0A} : p_{1j} = p_{2j} = p_{3j}, j=1,2 $, $H_{0AB} : p_{11} - p_{1\cdot} - p_{\cdot1} = \dots = p_{32} - p_{3\cdot} - p_{\cdot2}$, respectively.
Here, $p_{i\cdot} := \tfrac12(p_{i1} + p_{i2}), i=1,2,3$, and $p_{\cdot j}:= \tfrac13(p_{1j} + p_{2j} + p_{3j}), j=1,2$.
The equality of all uniform marginal distributions of $\vZ_k$ in combination with the affine-linear transformation to obtain $\vX_k$ guarantees that all of these null hypotheses are true.

Furthermore, we simulate missingness according to MCAR (Assumption~\ref{ass:MCAR}) by first generating data $\vy_1 = (y_{11}, \dots, y_{1d}), \dots, \vy_n= (y_{n1}, \dots, y_{nd})$ from an exchangeable Gaussian copula with Kendall's $\tau=0.2$ and then set $\lambda_{ik} := \ind\{y_{ik}> \Lambda_{i}\}$, $i=1, \dots,d, \ k=1,\dots,n$, where the expected missingness probabilities $\Lambda_i := P(\lambda_{i1} = 0)$ were chosen to be 
$$(\Lambda_1, \dots, \Lambda_{11})= (0.1, 0.1, 0.12, 0.19, 0.13, 0.18, 0.08, 0.1, 0.19, 0.11, 0.13) $$ 
$$ \text{and} \qquad  (\Lambda_{1},\dots, \Lambda_{6})= (0.1, 0.12, 0.19, 0.12, 0.08, 0.18),$$ respectively.
Overall, these choices lead to component-wise non-exchangeably and non-interchange\-ably distributed observable data. Consequently, none of the considered tests controls the type-I error probability for finite sample sizes.
Note that we have avoided to split the component index into two indices in the two-way layout in order to simplify the notation.

\subsection{Power simulation}
\label{sec:power}

We investigate the ability of the developed tests to detect deviations from the null hypothesis, by considering two kinds of scenarios under the alternative hypothesis.
First, in the one-factor setting, we consider a trend alternative obtained by
$\vX_k+\delta\cdot (0,1,\ldots,10)^\top/10$.
That is, all components except the first are shifted by an increasing proportion.
Next, in the two-factor setting, a one-component alternative is used, i.e.,
$\vX_k+\delta\cdot (1,0,0,0,0,0)^\top$.
For both kinds of alternatives, we use the same range of values, with $\delta\in\{0,0.15,0.3\}$, and each $\vX_k$ follows the distributions from Section~\ref{sec:type-I}.
Also, the same missingness probabilities are used.

\subsection{Results}
\label{sec:sim_res}
Some graphical summaries of the results of the simulation studies can be found in Figures~\ref{fig:MCAR_size_n}--\ref{fig:MCAR_power}; tables with detailed numerical results can be found in the appendix.
It is apparent that the test \(\Psi_n^\pi\), based on the quasi-randomization procedure, reliably controls the type-I error, even for very small sample sizes, especially in relation to the dimension $d \in\{6,11\}$; see Figure~\ref{fig:MCAR_size_n}.  In particular, it should be taken into account that  the bootstrap-based test $\Psi_n^\textnormal{boot}$ is slightly liberal. When available, the test $\Psi_n^{\pi\min}$, based on the reduced quasi-randomization scheme, exhibited empirical test sizes similar to those of $\Psi_n^\pi$.
The asymptotic test $\Psi_n^\textnormal{asy}$ has type-I errors off the charts (see Figure~\ref{fig:MCAR_size_n}).
As a consequence, it is not further considered in the subsequent figures and analyses.

The power of all three tests $\Psi_n^\pi, \Psi_n^{\pi\textnormal{min}}, \Psi_n^\textnormal{boot}$ is quite alike; see the detailed results in the appendix.
Thus, only the simulated power of $\Psi_n^\pi$ is displayed in Figure~\ref{fig:MCAR_power}. 
Not surprisingly, it rises with increasing sample sizes $n$ and shifts $\delta$.

\begin{figure}
\centering
\includegraphics[width=0.65\textwidth]{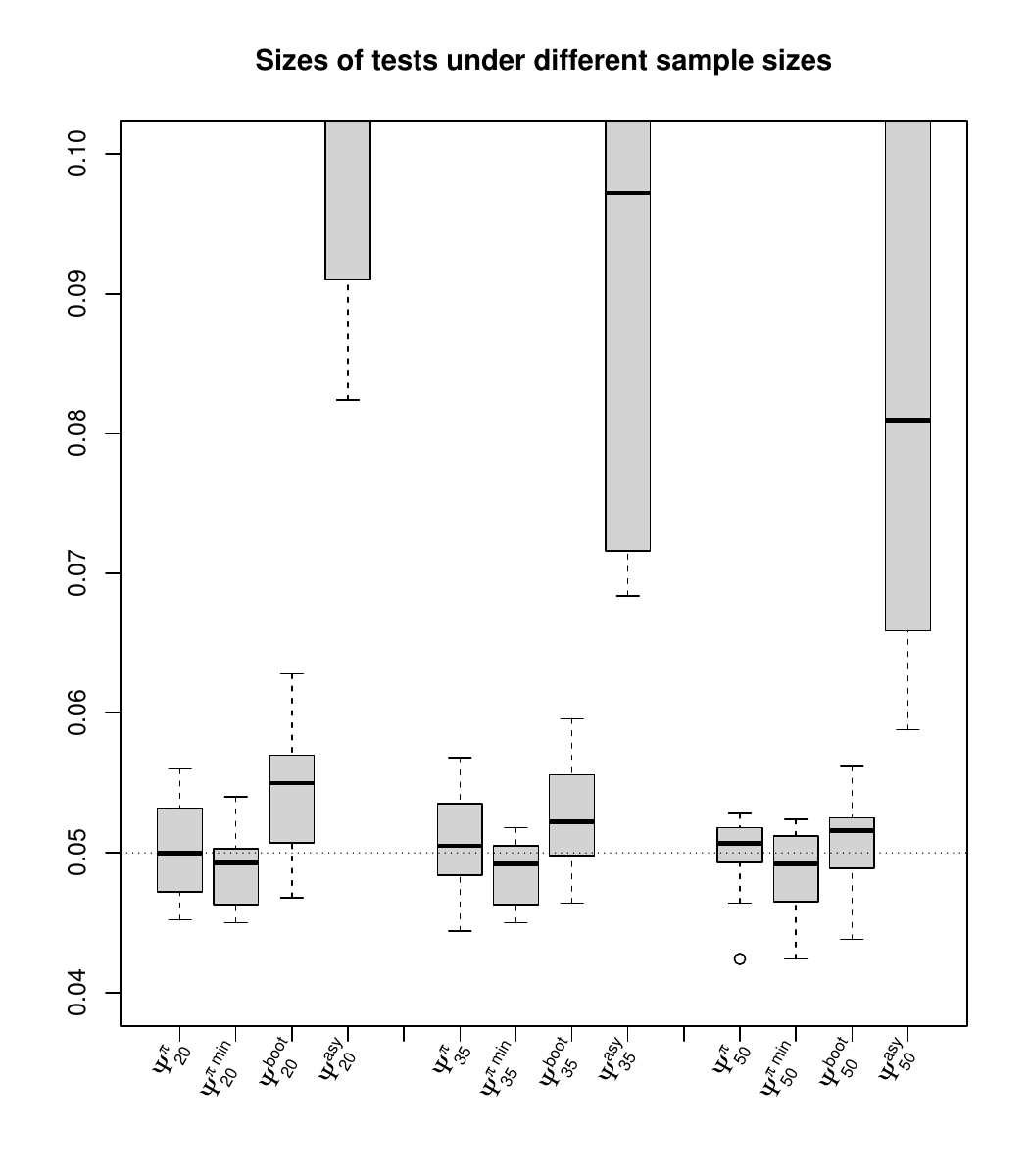}
\caption{Simulated sizes of all tests aggregated over all copulae and null hypotheses, under MCAR. 
The dotted line indicates the nominal significance level $\alpha=5\%$. Note that  $\Psi_n^{\pi\textnormal{min}}$ is not available for testing all hypotheses, so comparisons with this test demands caution.}
\label{fig:MCAR_size_n}
\end{figure}

\begin{figure}[h]
\centering
\includegraphics[width=0.65\textwidth]{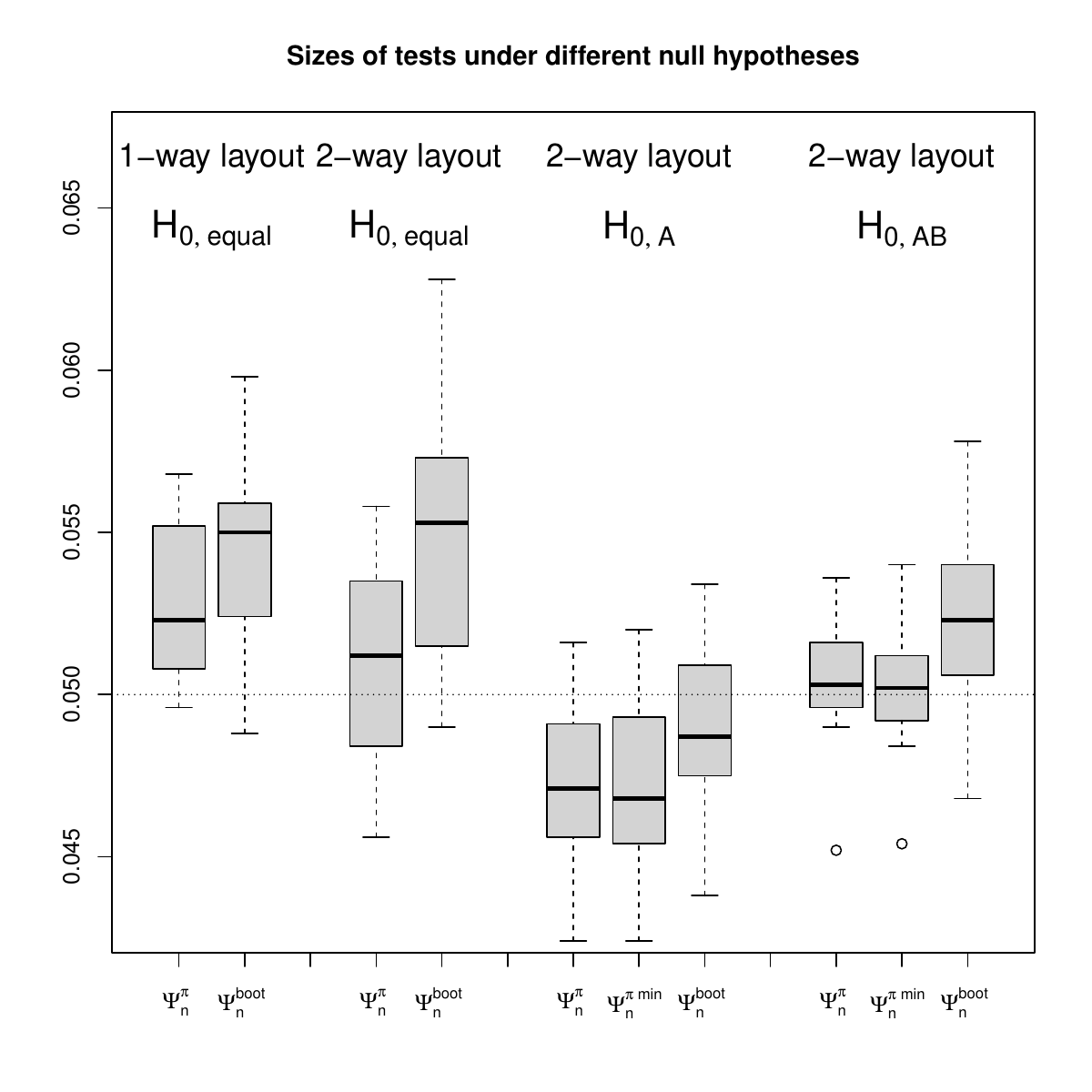}
\caption{Simulated sizes of all tests aggregated over all copulae and sample sizes, under MCAR. 
The dotted line indicates the nominal significance level $\alpha=5\%$. Note that  $\Psi_n^{\pi\textnormal{min}}$ is not available for testing $H_{0,\textnormal{equal}}$.}
\label{fig:MCAR_size_H0}
\end{figure}

\begin{figure}[h]
\centering
\includegraphics[width=0.45\textwidth]{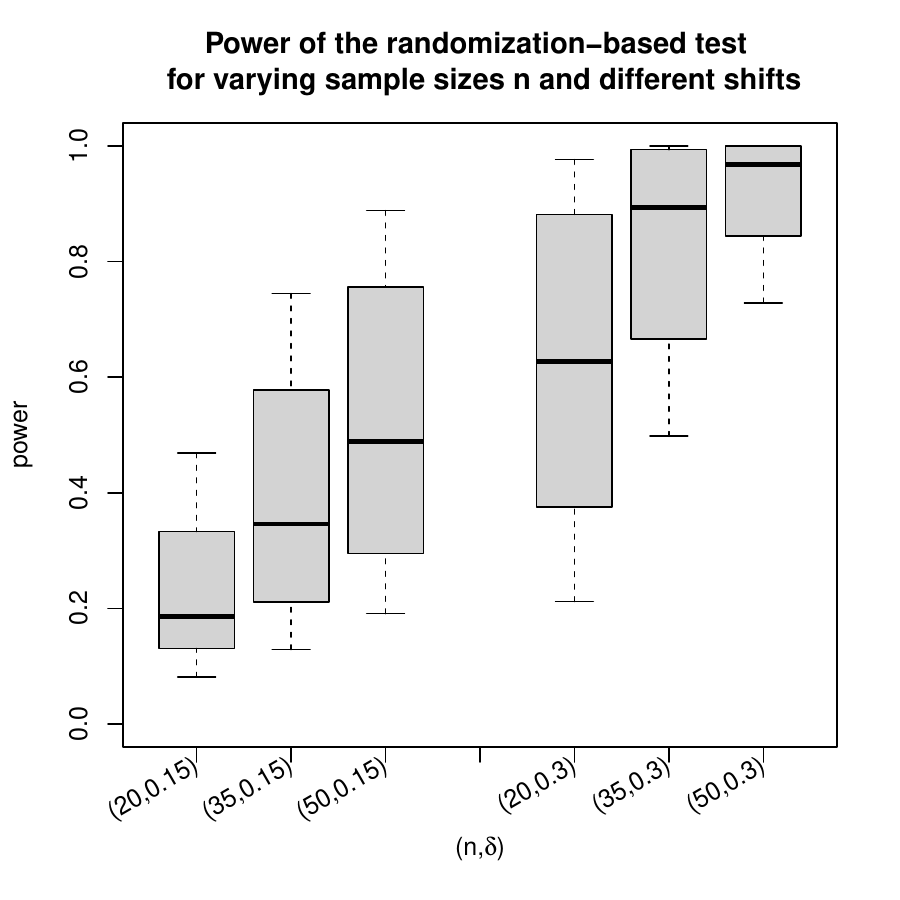}
\caption{Simulated power of the test $\Psi_n^\pi$ aggregated over all copulae, with varying sample sizes $n \in \{20,35,50\}$ (both of three adjacent boxplots) and shifts $\delta \in \{0.15, 0.3\}$ of the raw data, under MCAR.}
\label{fig:MCAR_power}
\end{figure}

\subsection{Additional simulation results}

The appendix contains additional simulation results that were obtained under the same settings, except that the missingness mechanism has been changed to `missing at random' (MAR).
Overall, the results are very similar to those presented in Section~\ref{sec:sim_res}, indicating that the developed approach is reasonably robust to violations of the MCAR assumption.

\subsection{Concluding remarks about the simulation studies}

All in all, we recommend the use of the tests based on quasi-randomization, which reliably control the type-I error probabilities.
If available, we also advise to use $\Psi_n^{\pi\min}$ instead of $\Psi_n^{\pi}$ because of the finite sample type-I error probability control of $\Psi_n^{\pi\min}$ for more data generating mechanisms.
There is hardly any difference in terms of power between these two tests.
The asymptotic test $\Psi_n^\textnormal{asy}$ should not be used under any circumstances because it is extremely liberal.

\section{Application to empirical data: Children learning math}\label{application}
The data of interest originate from a longitudinal intervention study about elementary school children with mathematical difficulties. 
These were defined as having a percentile rank of 25 or lower in a standardized math test \citep{roick2018demat3plus} at the beginning of the school year; see \cite{kuhn2019arithmetische} for details.
For the present paper, a random subsample (50\% of the total sample) was drawn. The random sample comprised 83 children: 40 children from third grade and 43 children from fourth grade, thereof, 34 boys and 49 girls.

Across the second school term, all children were administered a math-related learning progress assessment (LPA; \citealp{strathmann2012lvdm}) at 11 time points, with approximately two weeks between each assessment. The LPA was designed to assess basic and written calculation skills, with the first assessment taking place before the start of the interventions and thus serving as a baseline measurement.

The children were divided into three groups. The first group (termed `CODY') obtained a tablet-based math intervention during the second school term (25 children). It was designed specifically to address deficiencies of children with math difficulties, e.g., number and magnitude comparison \citep{kuhn2014number}. The math intervention comprised a maximum of 57 training days, each day requiring the child to work through two tasks for 10 minutes each. A second group (termed `CODY+NIP'; 12 children) obtained the same tablet-based math intervention, but additionally was administered a weekly small-group intervention across 12 sessions addressing magnitude and number concepts as well as basic calculation skills; the concept of this intervention is based on \cite{kaufmann2003numeracy}. Finally, a control group of 25 children did not receive any additional intervention during the second school term. 

Children were spread comparably across grades and groups ($\chi^2$(2) = 1.036, $p$-value = 0.596) as well as gender and groups ($\chi^2$(2) = 0.377, $p$-value = 0.828). Further, children differed in the degree of math difficulties: Children receiving the most intensive intervention (CODY+NIP) showed the lowest math $T$-scores ($M$ = 34.58, $SD$ = 1.78) compared to the CODY group ($M$ = 38.13, $SD$ = 3.82) and the control group ($M$ = 37.28, $SD$ = 5.00). Both intervention groups differed significantly from each other (Games-Howell $p$-value < 0.01); no other significant group differences were detected.

Subsample of size five from each intervention group are illustrated in Figure~\ref{fig:data}.
It is difficult to see a clear trend in these data, even in the case of observability, but on average the children seem to have a higher score in the second half of the tests compared to the first half.
In the following, we will analyze within each group whether there is any kind of learning effect.

\begin{figure}
    \centering
    \includegraphics[width=0.65\linewidth]{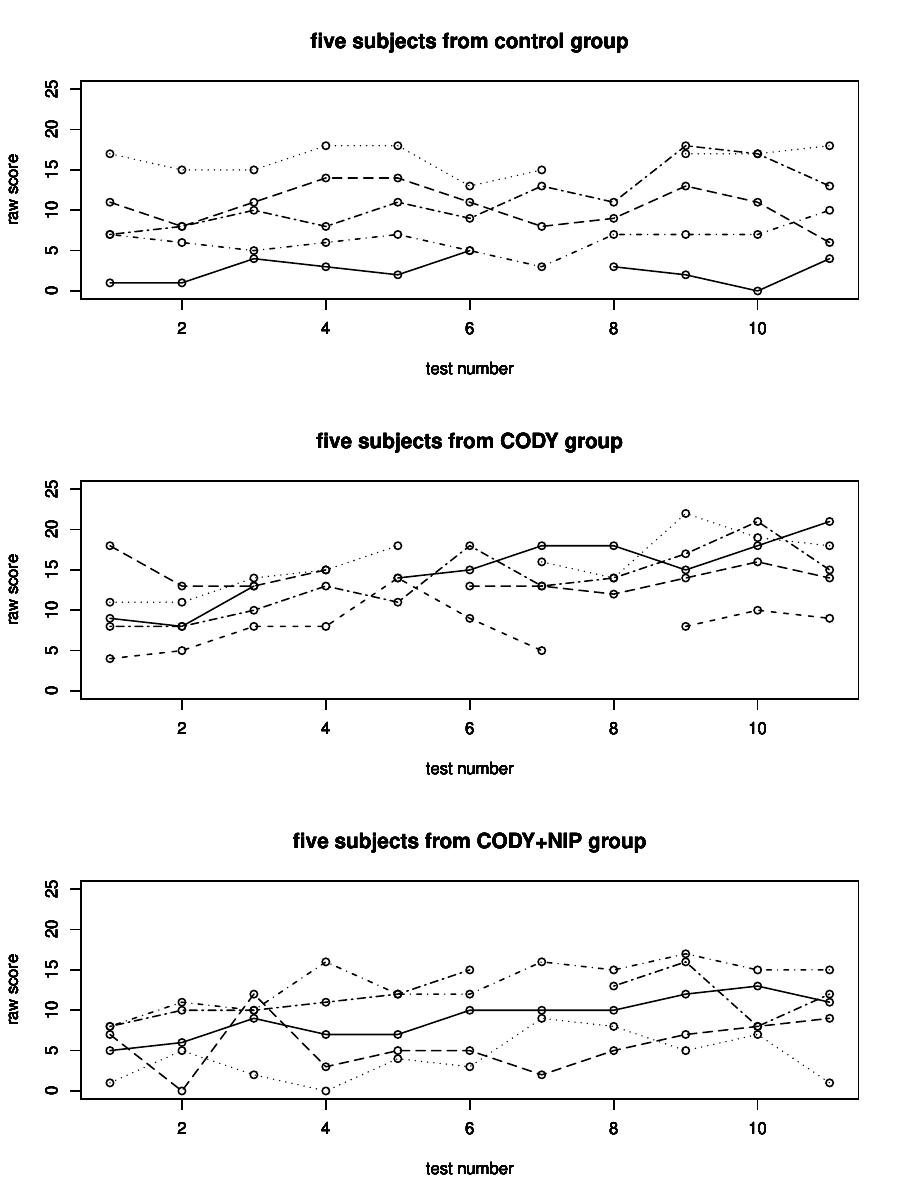}
    \caption{Illustration of the measurements of five individuals per intervention group. Observable scores are marked with a circle. Two adjacent circles are connected with a line when there is no missingness. In each group, different children are illustrated with different line types.}
    \label{fig:data}
\end{figure}

Thus, to investigate whether the interventions have a significant influence on mathematical performance, we test for each of the three study groups the relevance of the factor `time'.
To this end, we test the hypotheses 
$$H_{0,\textnormal{equal}}: p_1 = \dots = p_{11} \qquad \text{vs.} \qquad  H_{1,\textnormal{equal}}:  p_i \neq p_j \ \ \text{for some} \ \ i \neq j$$ in the context of a one-way factorial design.

\textcolor{black}{To this end, we applied the quasi-randomization test considered in the simulation study, using $2{,}000$ resampling repetitions. The resulting $p$-values are $0.004$ (control), $0.012$ (CODY), and $0.0045$ (CODY+NIP).
Thus, the tests for equality of all Mann-Whitney effects are significant at level $\alpha=5\%$ for all three intervention groups.}

In addition, to better understand the behavior over time, Figure~\ref{fig:groupwise_confidence_intervals} contains an illustration of all point estimates and two-sided Wald-type level-$95\%$-confidence intervals from Remark~\ref{rem:CIrandom} for the relative effects at each test number (`component'). These confidence intervals are intended only to illustrate the development of the relative effects over time. They are not adjusted for simultaneous inference and, therefore, should be interpreted in a point-wise manner rather than time-simultaneously.
It is apparent that, in the fifth test, the point estimate in the CODY group is the first to cross the 0.5-threshold; nearly all subsequent point estimates in this group remain beyond that line, indicating a robust learning success. 
The CODY+NIP group exhibited the steepest increase between Test~4 and Test~9. The point estimates of the control group often lie between those of the other two groups. 
However, it is important to note that inter-group comparisons are not necessarily meaningful because the Mann-Whitney effects of different groups might be similar even when the original data are on completely different levels, e.g., when all values of a group are simply shifted by the same value. Nevertheless, the \emph{increase} of Mann-Whitney effects provide important insights about the timeliness of intervention effects.
In this sense, the CODY group seems to have the greatest increase, namely between Tests~2 and~3, which is quite early.
It is also interesting to see that all derived confidence intervals for Tests~10 and~11 are beyond the 0.5-line, even for the control group.
This shows that all groups experienced a learning effect in the long run.
\begin{figure}[htbp]
\centering
\includegraphics[width=0.98\textwidth]{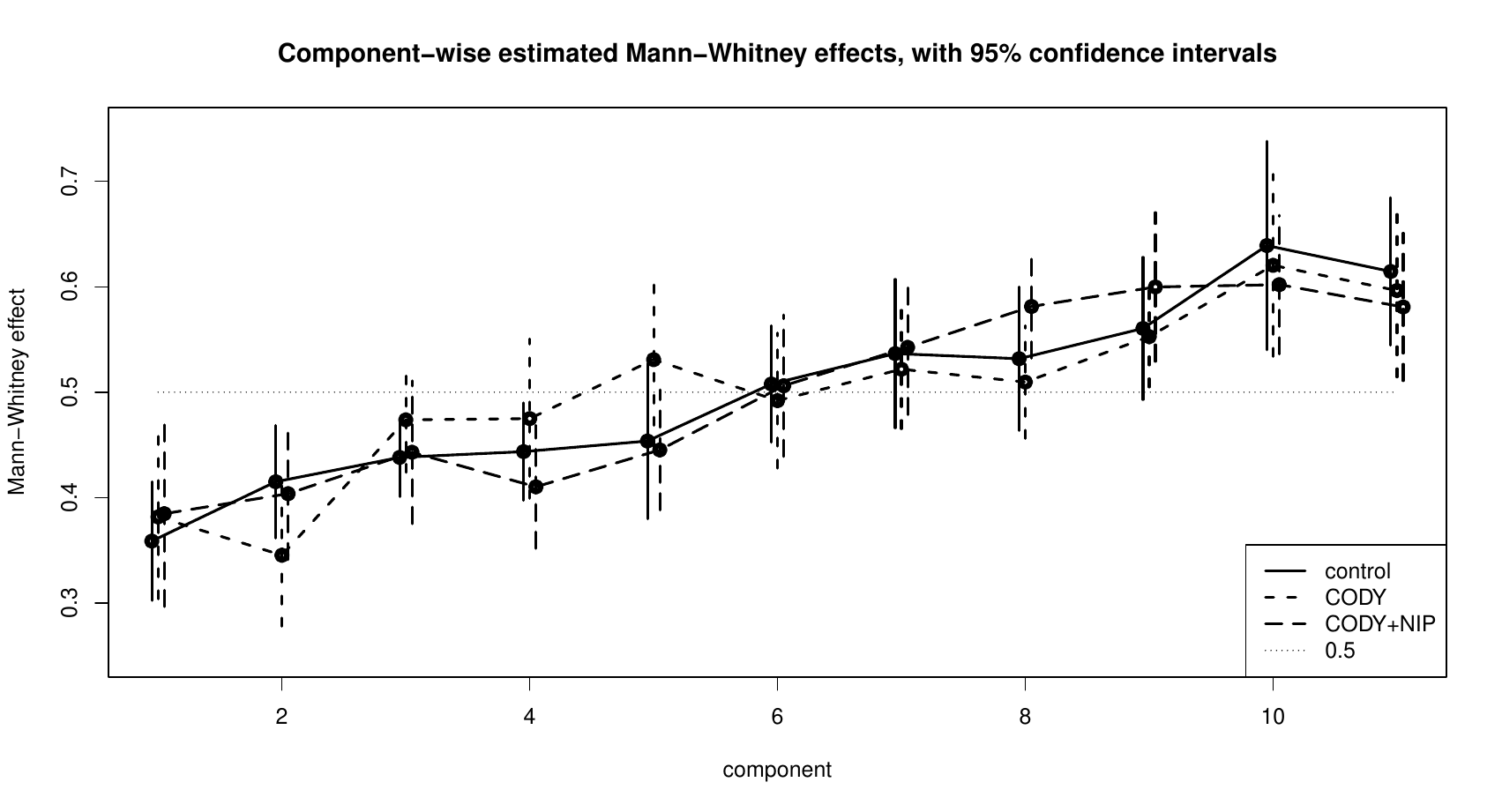}
\caption{Component-wise estimated Mann-Whitney effects with two-sided point-wise 95\% confidence intervals at the $11$ time points for the three study groups. The confidence intervals are not adjusted for simultaneous inference. For ease of presentation, the three groups are slightly shifted on the horizonal axis.}
\label{fig:groupwise_confidence_intervals}
\end{figure}

\section{Conclusion}\label{conclusion}
Although missing values can occur in many types of data for a wide variety of reasons, to date relatively few methods have been developed  for ordinal data in this context under realistic missing-data mechanisms. 
The present work requires the MCAR framework but allows for quite general patterns, e.g., different component-specific missingness proportions and correlations of the missingness indicators within each data point.
By employing a quasi-randomization technique we developed asymptotically valid tests. Importantly, the distribution and frequency of missing values naturally enter the covariance structure in this framework and are fully accounted for. An estimator of the covariance matrix is used in the Wald-type test statistic and, based on the quasi-randomized data, also in the quasi-randomized test statistic.
This empirical-type estimator of the covariance matrix is based on the asymptotic influence function, it leads to reliable results, and it could be implemented in a comparatively efficient manner. 

Extensive simulations across a wide range of settings demonstrate that the proposed method perform well even for small sample sizes and compares favorably with alternative approaches, e.g., the bootstrap, highlighting its robustness and practical applicability.

A popular alternative approach for dealing with missing data, such as under the MCAR or the MAR regime, might involve multiply imputing missing entries. A widely-used missing data method in practice, MICE (multivariate imputation by chained equations; \citealp{vanbuuren2011}), conducts multivariate imputation on a variable-by-variable basis using conditional densities, iteratively drawing missing values starting from an initial imputation. However, analyzing the covariance structure of the resulting point estimators and deducing inference procedures would not be trivial; this can be considered in future research as a competitor method after adjusting the present methodology for the MAR assumption which is of considerable practical interest.

Future work will also focus on extending these results to a multiple-groups or multivariate setting, as examined, for example, in \cite{rubarth22a}. Such an extension would allow the analysis of split-plot designs, enabling a more comprehensive assessment of factor effects and interactions in multivariate outcomes. 
In the context of the present real data problem, multiple sample comparisons could be used to conclude which intervention led to the greatest learning effect.
The methodological extension will likely require techniques beyond pure quasi-randomization.

\textcolor{black}{Other future research includes a thorough statistical analysis of estimators and corresponding resampling schemes for $(\theta_1, \dots, \theta_d)$ from Remark~\ref{rem:fay} and comparisons with $(p_1,\dots,p_d)$.
It will be interesting to investigate multivariate variants of Hand's paradox; see \cite{fay18} for details in the two-dimensional case.}
Lastly, we also plan to develop multivariate one-sided tests as well as regression techniques based on Mann-Whitney effects.

\section*{Acknowledgements}
The authors would like to thank Markus Pauly for discussions regarding the choice of this topic, Edgar Brunner and Stefanie Schoenen for helpful discussions, and Nils Hichert for support with the implementation of the methods in the statistical software \textit{R}. Dennis Dobler wishes to thank his former affiliations, TU Dortmund University and Research Center Trustworthy Data Science and Security where the work was initiated.
Paavo Sattler would also like to thank his other affiliation, the Department of Statistics at TU Dortmund University.

\section*{Declaration of generative AI use}
We have used generative AI for some internal mathematical discussions, some few coding tasks, for support in commenting the R code, and creating the tables with simulation results.

\section*{Data availability}
The R code for running the developed methods on the CODY data and also on a toy data set can be found in an online repository: \\
\texttt{https://github.com/dennis-dobler/Mann\_Whitney\_effects\_randomization\_test}

\newpage

\appendix

\section*{Appendix}

Next to all proofs, Appendix~\ref{app:proofs} also contains derivations of the asymptotic ranks of the relevant covariance matrices after the (reduced) quasi-randomization schemes. 
Explicit formulas for the asymptotic covariances and consistent estimators for these are provided in Appendix~\ref{app:covariance}.
The concluding Appendix~\ref{app:rankexample} offers an example of a sampling distribution with reduced rank covariance matrix where the corresponding asymptotic covariance matrix of $\hat{\vp}$ has the almost full rank $d-1$. Appendix~\ref{app:add_sim} contains the detailed simulation results a summary of which has been presented in Section~\ref{sec:sim_res} in the main body of the paper.
Finally, Appendix~\ref{app:MAR} contains additional simulation results under the weaker MAR assumption.

\setcounter{section}{0}

\section{Proofs}
\label{app:proofs}

\subsection{Proofs related to the estimators' asymptotics}

\phantom{X}

\medskip

\noindent \textbf{Proof of Proposition~\ref{prop:bias}.}
Let us consider the case of no observations in a component $i\in \{1,\dots,d\}$ first.
Technically, this can be realized by extending $\mathbb{R}$ to $\overline{\mathbb{R}} := \mathbb{R}\cup \{-\infty, \infty\}$ and by using the symmetric two-point probability measure with mass in $-\infty$ and $\infty$ instead: $\hat F_i(x) := \tfrac12 (\delta_{-\infty} + \delta_{\infty})([-\infty, x]) \equiv \tfrac12$ for $x \in \mathbb R$. 
The central advantage of this definition is that no arbitrary choice for the value of $\hat p_i$ needs to made since $\hat p_i= \tfrac12$ follows; also, note that $\int \hat F_i d \hat F_i = \tfrac12$ holds in that case.

Denote by $\mathcal{A}_{\lambda}$ the $\sigma$-algebra generated by all $\lambda_{ik}, i=1, \dots, d, \ k=1, \dots, n$. Then, defining $0/0:= 0$,
    \begin{align*}
        E(\hat F_i(x) \ | \ \mathcal{A}_{\lambda}) = \ind\{\lambda_{i \cdot} > 0 \} \frac{1}{\lambda_{i.}}\sum_{k=1}^{n}\lambda_{ik} E( c(x-X_{ik}) ) + \ind\{\lambda_{i \cdot} = 0 \} \Big(\frac12 \ind_{\mathbb R}(x)  + \ind\{x = \infty\} \Big).
    \end{align*}
    which follows from the MCAR Assumption~\ref{ass:MCAR}.
    Thus, 
    $$E(\hat F_i(x)) = P( \lambda_{i\cdot} > 0 ) \cdot F_i(x) + P(\lambda_{i \cdot} = 0 \} \big(\tfrac12 \ind_{\mathbb R}(x)  + \ind\{x = \infty\} \big).$$ A similar result can be shown for $\hat H$ which is a convex combination of $\hat F_1, \dots, \hat F_d$.

    Using similar arguments and the fact that $\int \hat F_i d \hat F_i = \tfrac12$,
    it is straightforward to see that
\begin{align*}
    E(\hat p_i\ | \ \mathcal{A}_{\lambda}) & = E\Big( \int \hat H d \hat F_i \ | \ \mathcal{A}_{\lambda}\Big) = \frac1d  E\Big(\int \hat F_i d \hat F_i \ | \ \mathcal{A}_{\lambda}\Big)
    + \frac1d  E\Big( \sum_{s\neq i} \int \hat F_s d \hat F_i \ | \ \mathcal{A}_{\lambda} \Big) \\
    & = \frac1{2d} + \frac1d\sum_{s\neq i} \Big( \frac12 \ind\{\lambda_{s\cdot} = 0 \text{ or } \lambda_{i\cdot} = 0 \}  + \ind\{\lambda_{s\cdot} > 0, \lambda_{i\cdot} > 0 \} \int F_s d F_i \Big) + \Lan(n^{-1})
\end{align*}
and thus, upon taking the missingness indicators out of the conditional expectation and then deriving the unconditional expectation,
    \begin{align*}
        E(\hat p_i ) & = \frac1{2d} + \frac1d \sum_{s \neq i} P\Big(\sum_{k=1}^n \lambda_{sk} > 0 , \sum_{k=1}^n \lambda_{ik} > 0 \Big) \int F_s d F_i  +  \Lan(n^{-1}) \\
        & \to  \frac1d \Big( \frac12 + \sum_{s \neq i} \Big( P(X_{i1} > X_{s2}) + \frac12 P(X_{i1} = X_{s2}) \Big)\Big)  = E(H(X_{i1})) = p_i
    \end{align*}
    Here, $\Lan(n^{-1}) = E ( \lambda_{i \cdot}^{-1} \lambda_{s \cdot}^{-1}\sum_{k=1}^n \lambda_{ik} \lambda_{sk} ) (P(X_{i1} \succeq X_{s1})  - P(X_{i1} \succeq X_{s2}))$.
\hfill\qed

\bigskip

Our analysis of the asymptotic normality of the estimator $\hat {\vp}$ begins with a its representation as a Hadamard-differentiable functional $\psi: \ell^\infty(\mathcal{F}) \to \mathbb{R}^d$ of the empirical measure, where  $\mathcal{F}$ is a suitable indexing set.
The estimand $\vp$ may be represented with the help of the same functional.
Thus, 
$\hat {\boldsymbol{p}}= \psi(\mathbb{P}_n)$ and $ {\boldsymbol{p}}= \psi(\mathbb{P})$ where
$\mathbb{P}_n = \frac1n \sum_{k=1}^n \delta_{(\vX_k, \vlam_k)}$ denotes the empirical process of all $2d$-dimensional data points $(\boldsymbol{X}_k, \boldsymbol{\lambda}_k), \ k=1,\dots, n$, and $\mathbb{P} = P^{(\boldsymbol{X}_k, \boldsymbol{\lambda}_k)}$ is the distribution of each random vector representing a data point.

In order to prepare the proofs with the help of empirical process theory, we introduce some more notation:
Let $D^\uparrow(\mathbb R)$ be the space of normalized cumulative distribution functions on $\mathbb R$.
We will see that, without loss of generality, it is sufficient to consider this space in our arguments involving functionals.
Otherwise, one could represent all considered Mann-Whitney effect functionals as a sum of two terms -- one of which involves the right-continuous distribution functions and the other one involves the left-continuous variants in the integrand; each space (of left- or right-continuous functions) is equipped with the supremum norm.
Hence, such a treatment would require an extension of the spaces but poses no technical difficulties.

 Let $\mathcal{F} := \{ f_{ix} = (c_x \circ pr_{i1}^d) \cdot pr_{i2}^d , \ pr_{i2} : \ x \in \R, \ i = 1,\dots, d \}$
    where 
    \begin{align*}
    c_x: & \quad  \R \to \{0,\tfrac12,1\},  && z \mapsto c(x-z) , \qquad \text{and} \\
        pr_{i1} : &\quad  \mathbb R^d \times \{0,1\}^d \to \mathbb R, && (x_1,\dots, x_d, \ell_1, \dots, \ell_d) \mapsto x_i, \\
        pr_{i2} : &\quad  \mathbb R^d \times \{0,1\}^d \to \{0,1\}, &&  (x_1,\dots, x_d, \ell_1, \dots, \ell_d) \mapsto \ell_i,
    \end{align*}
    are the canonical projections.
    In the following, product spaces will always be equipped with the max-sup norm.
    We are going to use the functional notation $\P f := \int f d \P$.
\begin{lem}
\label{lemma:psi}
   \begin{itemize}
       \item[\textnormal{(a)}] $\mathcal{F}$ is a Donsker class.
       \item[\textnormal{(b)}]  The functional
       $\psi :  \ell^\infty(\mathcal{F})  \to \big(D^\uparrow(\mathbb R)\big)^d \times \R^d  \to  \big(D^\uparrow(\mathbb R)\big)^{d+1}  \to  \mathbb{R}^d$, 
     \\[0.1cm]
     $\mathbb P \mapsto  (\check F_1, \dots, \check F_d, \kappa_1, \dots, \kappa_d)  \mapsto \Big(\frac{\check F_1}{\kappa_1}, \dots, \frac{\check{F}_d}{\kappa_d}, \frac1d \sum_{i=1}^d \frac{\check F_i}{\kappa_i }\Big)  \mapsto  ( \int \frac1d \sum_{i=1}^d \frac{\check F_i}{\kappa_i} \frac{d \check F_s}{\kappa_s} : s= 1,\dots, d )$
     \\[0.1cm]
    is Hadamard-differentiable at every $\P$ for which 
    \begin{itemize}
        \item[(i)] (identifiability) $\kappa_i := \P pr_{i2} = P(\lambda_{i1} = 1) > 0$, $i=1,\dots, d$;
        \item[(ii)] (MCAR) $ \check F_i(x) := \P (c_x \circ pr_{i1}) \cdot  pr_{i2} = P(X_{i1} < x, \lambda_{i1} = 1) +  \tfrac12 P(X_{i1} = x, \lambda_{i1} = 1) = \kappa_i F_i(x)$, $i=1,\dots, d$, \ $x\in \R$;
    \end{itemize}
    tangentially to $D_0(\R)$, the space of normalized functions $h$ on $\R$ such that $\lim_{|x| \to \infty} h(x) = 0$.
    The Hadamard-derivative is defined as
    \begin{align*}
        & \psi'_{(\check F_1, \dots, \check F_d, \kappa_1, \dots, \kappa_d)} \cdot (h_1, \dots, h_d, \ell_1, \dots, \ell_d) \\
        & =  \Big( \frac1d \sum_{s=1}^d \Big[ \int \Big(\frac{h_s}{\kappa_s} - \frac{\ell_s \check F_s}{\kappa_s^2}\Big) d \frac{\check F_i}{\kappa_i} - \int \Big(\frac{h_i}{\kappa_i} - \frac{\ell_i \check F_i}{\kappa_i^2} \Big) d \frac{\check F_s}{\kappa_s} \Big] : i = 1, \dots, d \Big) \\
        & =  \Big( \frac1d \sum_{s=1}^d \Big[\kappa_s^{-1} \int \Big(h_s - \ell_s F_s\Big) d  F_i - \kappa_i^{-1}\int \Big(h_i - \ell_i F_i \Big) d F_s \Big] : i = 1, \dots, d \Big) ;
    \end{align*}
    the last equality is owing to $\check F_i / \kappa_i = F_i, i=1,\dots,d$, under the MCAR assumption.
   \end{itemize}
\end{lem}
\noindent \textbf{Proof of Lemma~\ref{lemma:psi}.}
(a) It is clear that $\mathcal{F} $ is a Vapnik-\u{C}ervonenkis (VC) class and hence a Donsker class because the used operations (concatenation, product, and union) preserve these properties; cf.\ Chapter~2 of \cite{vdVW23}. 

(b)
The proof follows from the differentiability result for the fraction functional and the Wilcoxon functional combined with integration by parts for the derivative formula; see Lemma~3.10.18 in \cite{vdVW23} and the Supplementary Material of \cite{dobler2020b} for the Hadamard-differentiability of a similar functional.
\hfill \qed

\bigskip

In the remainder, we will use the shorter notation $\psi'_{\P} := \psi'_{(\check F_1, \dots, \check F_d, \kappa_1, \dots, \kappa_d)}$
and also perceive this to be a functional of the empirical process rather than a functional defined on $\big(D^\uparrow(\mathbb R)\big)^d \times \R^d$; the presentation in the previous lemma was chosen for reasons of legibility.

\bigskip
\noindent
\textbf{Proof of Theorem~\ref{thm:clt}.}
Based on the results from Lemma~\ref{lemma:psi} and the functional delta-method, e.g., Theorem~3.10.4 in \cite{vdVW23}, we use the following asymptotic linearization of the estimators, 
\begin{align*}
    \displaystyle
 \sqrt{n}(\hat {\boldsymbol{p}} - \boldsymbol{p}) = \psi'_{\mathbb{P}} (\sqrt{n}(\mathbb{P}_n - \mathbb{P})) + \lan_p(1)
= \frac1{\sqrt{n}} \sum\limits_{k=1}^n ( \dot\psi_{\mathbb{P}} (\boldsymbol{X}_k, \boldsymbol{\lambda}_k) - E(\dot\psi_{\mathbb{P}} (\boldsymbol{X}_k, \boldsymbol{\lambda}_k))) + \lan_p(1).
\end{align*}
Here, $\dot\psi_{\mathbb{P}}(\vX_k, \vlam_k) :=  \psi'_{\P}(\delta_{(\vX_k, \vlam_k)}) $ denotes the influence function corresponding to $\psi$, 
where $\delta_{(\vX_k, \vlam_k)}$ is the Dirac measure in $(\vX_k, \vlam_k)$.
We conclude the claimed central limit theorem and also the following asymptotically linear structure for $\sqrt{n}(\hat {\boldsymbol{p}} - \boldsymbol{p})$:
\begin{align}
\begin{split}
\label{eq:asy_lin}
    & \frac1{\sqrt{n}d} \sum_{k=1}^n \sum_{s=1}^d   \Big( \Big[  \int \Big\{ \frac{(\ind\{X_{sk} < x\} + \tfrac12 \ind\{X_{sk} = x\}) \lambda_{sk}}{\kappa_s} - \frac{\check F_s(x) \lambda_{sk}}{\kappa_s^2} \Big\} d  F_i(x) \\
    & \ \ - \int\Big\{ \frac{(\ind\{X_{ik} < x\} + \tfrac12 \ind\{X_{ik} = x\}) \lambda_{ik}}{\kappa_i} - \frac{\check F_i(x) \lambda_{ik}}{\kappa_i^2} \Big\} d   F_s(x) \Big] : i = 1 ,\dots, d \Big) + \lan_p(1) \\
    & =  \frac1{\sqrt{n}d} \sum_{k=1}^n \sum_{s=1}^d   \Big( \Big[  \frac{\lambda_{sk}}{\kappa_s} \int \Big\{ \ind\{X_{sk} < x\} + \tfrac12 \ind\{X_{sk} = x\} -  F_s(x) \Big\} d  F_i(x) \\
    & \ \  - \frac{\lambda_{ik}}{\kappa_i}\int\Big\{ \ind\{X_{ik} < x\} + \tfrac12 \ind\{X_{ik} = x\}  - F_i(x)  \Big\} d   F_s(x) \Big] : i = 1 ,\dots, d \Big) + \lan_p(1)
    \\
     & =  \frac1{\sqrt{n}d} \sum_{s=1}^d   \Big( \Big[  \frac{\lambda_{s\cdot}}{\kappa_s} \int \Big\{ \hat F_s -  F_s \Big\} d  F_i - \frac{\lambda_{i\cdot}}{\kappa_i}\int\Big\{ \hat F_i - F_i  \Big\} d   F_s \Big] : i = 1 ,\dots, d \Big) + \lan_p(1)
\end{split}
\end{align}
where the second to last equality follows from the MCAR assumption; cf.\ Assumption~\ref{ass:MCAR}.
\hfill \qedsymbol

\bigskip

Note that $E(\dot\psi_{\mathbb{P}} (\boldsymbol{X}_k, \boldsymbol{\lambda}_k))) = \boldsymbol{0}_d$ which will be of great importance below.

\color{black}

\bigskip

\noindent
\textbf{Proof of Theorem~\ref{thm:T}.} By the Donsker theorem for $\mathbb{P}_n$ in combination with the functional $\delta$-method, 
it follows under the null hypothesis that
$\sqrt{n}(\vC\hat{\boldsymbol{p}} - \vc)= \vC\cdot \sqrt{n}(\hat{\boldsymbol{p}} - \boldsymbol{p}) \stackrel d\to \boldsymbol{C Z}$, 
where  $\boldsymbol Z \sim \mc N_d(\boldsymbol{0}_d, \boldsymbol{V})$.

As mentioned in the main body, we may derive from Assumptions~\ref{ass:test} and~\ref{ass:V} that $P(\mathrm{rk}(\hat{\boldsymbol{V}}) = d-1) \to 1$ as $n\to \infty$.
Similarly, we establish that
$$P(\mathrm{rk}(\boldsymbol{C} \hat{\boldsymbol{V}}) = \mathrm{rk}(\boldsymbol{C} {\boldsymbol{V}})) = P(\mathrm{rk}(\boldsymbol{C} \hat{\boldsymbol{V}} \boldsymbol{C}^\top) = \mathrm{rk}(\boldsymbol{C} {\boldsymbol{V}} \boldsymbol{C}^\top)) \to 1$$ due to $\boldsymbol{C} \hat{\boldsymbol{V}} \boldsymbol{C}^\top \stackrel p \to \boldsymbol{C} {\boldsymbol{V}} \boldsymbol{C}^\top.$
Using this together with Corollary 1.8 from \cite{koliha2001mpinverse} it follows that
$(\vC\hat{\boldsymbol V}\vC^\top)^{+} \stackrel p\to (\vC{\boldsymbol V}\vC^\top)^{+}$ and hence, $$[(\vC\hat{\boldsymbol V}\vC^\top)^{+}]^{1/2} \sqrt{n}(\vC\hat{\boldsymbol{p}} - \vc) \stackrel d\to \mc N_{d}(\boldsymbol{0}, [(\vC{\boldsymbol V}\vC^\top)^{+}]^{1/2} (\vC{\boldsymbol V}\vC^\top)[(\vC{\boldsymbol V}\vC^\top)^{+}]^{1/2}).$$
Since this limiting variance-covariance matrix is idempotent, we can use Corollary 5.1.2a in \cite{Mathai1992}. Consequently, due to 
$$\mathrm{rk}([(\vC{\boldsymbol V}\vC^\top)^{+}]^{1/2} (\vC{\boldsymbol V}\vC^\top)[(\vC{\boldsymbol V}\vC^\top)^{+}]^{1/2})= \mathrm{rk}(\vC \boldsymbol{V} \vC^\top)=\mathrm{rk}(\boldsymbol{CV}), $$
the asymptotic distribution of $T_n$ is $\chi^2_{\mathrm{rk}(\boldsymbol{CV})}$. 
Under the alternative hypothesis,
\begin{align*}
    T_n & = \sqrt{n} \Big(\sqrt{n} (\hat{\boldsymbol{p}} - \boldsymbol{ p})^\top \vC^\top (\vC\hat{\boldsymbol V}\vC)^{+} \vC(\hat{\boldsymbol{p}}- \boldsymbol{p}) \\
    & \quad 
    + 2 \sqrt{n} (\hat{\boldsymbol{p}} - \boldsymbol{ p})^\top \vC^\top (\vC\hat{\boldsymbol V}\vC^\top)^{+} (\vC \boldsymbol{p} - \vc)
    \\
    & \quad + \sqrt{n}  (\vC \boldsymbol{p} - \vc)^\top (\vC\hat{\boldsymbol V}\vC^\top)^{+} (\vC  \boldsymbol{p} - \vc) \Big)
\end{align*}
  which is equal to $\sqrt{n} ( \lan_p(1) + \Lan_p(1) + \Lan_p(\sqrt{n})) \stackrel p \to \infty $ as $n \to \infty$.
\hfill$\qed$

\subsection{Proofs of the quasi-randomization-based statements}

\begin{lem}
\label{lem:rand_EP}
    Conditionally on $\mathcal{A}_{X,\lambda}$,
the following weak convergence holds in probability as $n \to\infty$:
$$ \sqrt{n} (\tilde{\mathbb P}_n - \mathbb{P}_n^{\mathcal{S}_d}) \rightsquigarrow \tilde {\mathbb{G}} \quad \text{in } \ell^\infty(\mathcal{F})  $$
where $\mathbb{P}_n^{\mathcal{S}_d}$ is the conditional expectation of  $\tilde{\mathbb P}_n$ given $\mathcal{A}_{X,\lambda}$, and $\tilde{\mathbb{G}}$ is a  zero-mean Gaussian process.
\end{lem}
\noindent
\textbf{Proof of Lemma~\ref{lem:rand_EP}.}
The set of indexing functions $\mathcal{F}$ from the previous proof is a VC-class, and it satisfies all required Donsker properties as stated Theorem~1 of \cite{dobler23}.
Thus, it follows that the randomization empirical process satisfies a conditional central limit theorem as claimed.
\hfill\qed

\bigskip

The following representation in terms of functionals of the randomization empirical process $\tilde{\mathbb P}_n$ of $(\vX_{k}^\pi, \boldsymbol \lambda_{k}^\pi)$, $k=1,\dots,n$, will prove useful:
    $\hat {\boldsymbol{p}}^\pi  = \psi(\tilde{\mathbb P}_n) $. Here, 
    $\psi$ is a Hadamard-differentiable functional with Hadamard-derivative 
    
    $\psi'_{\tilde{\mathbb{P}}}$ 
    which is a 
    continuous and linear functional on the space $\ell^{\infty}(\mathcal{F})$. 
    This is the same functional that also defines $\hat{\vp}$ and $\vp$ when evaluated at $\mathbb{P}_n$ and $\mathbb{P}$, respectively.

\bigskip

\noindent
\textbf{Proof of Theorem~\ref{thm:main_rand}.}
The functional $\psi$ that leads to the Mann-Whitney effect estimates is essentially a combination of summation, division, and the Wilcoxon functional $\phi: D^{\uparrow}_1(\mathbb R) \times D^{\uparrow}_1(\mathbb R) \rightarrow \mathbb R,  (f,g) \mapsto \int f(x) d g(x)$.
 Also,  the set of functions $\mathcal{F}$ is exactly the indexing class that is required to represent $\hat{\boldsymbol{p}}$ and $\hat{\boldsymbol{p}}^\pi$ as a Hadamard-differentiable functional of $\mathbb{P}_n$ and $\tilde{\mathbb{P}}_n$, respectively.

Since the missingness proportions are assumed to be contained in the interval $(0,1)$, all involved functionals and hence also their combination satisfy the required Hadamard-differentiability assumptions of the functional delta-method in Theorem~2 of \cite{dobler23}.
This result is combined with the convergence result of Lemma~\ref{lem:rand_EP} to conclude the proof for the claimed convergence;
the limit is a functional of $\tilde{\mathbb{G}}$, a Gaussian process in $\ell^\infty(\mathcal{F})$: $ \tilde\vZ :=\psi'_{\tilde{\mathbb{P}}}(\tilde{\mathbb{G}}) \sim \mathcal{N}_d(\boldsymbol{0}_d, \tilde{\vV}) $ is an exchangeable and multivariate normally distributed random vector such that $\boldsymbol{1}^\top_d \tilde \vZ = 0$.
Thus, the fact that $ \tilde \rho = -\tilde \sigma^2/(d-1)$ follows from $\boldsymbol{1}_d^\top \tilde {\vV}\boldsymbol{1}_d = 0$.
Furthermore, to establish $\mathrm{rk}( \tilde{\vV} ) = d-1$, we use that
\[
\boldsymbol{x}^\top \tilde{\vV} \boldsymbol{x}
=
\frac{\tilde\sigma^2}{d-1}
\left(
d \|\boldsymbol{x}\|^2
-
(\boldsymbol{1}_d^\top \boldsymbol{x})^2
\right), \quad  \boldsymbol{x} \in \mathbb{R}^d.
\]
By the Cauchy--Schwarz inequality, $(\boldsymbol{1}_d^\top \boldsymbol{x})^2 \le d \|\boldsymbol{x}\|^2$, with equality if and only if $\boldsymbol{x} \in \mathrm{span}(\boldsymbol{1}_d)$. Hence, $\boldsymbol{x}^\top \tilde{\vV} \boldsymbol{x} \ge 0$, and equality holds if and only if $\boldsymbol{x} \in \mathrm{span}(\boldsymbol{1}_d)$. This shows that $\mathrm{ker}(\tilde{\vV}) = \mathrm{span}(\boldsymbol{1}_d)$, and thus $\mathrm{rk}(\tilde{\vV}) = d-1$ if $\tilde \sigma^2 > 0$.

To appreciate why $\tilde \sigma^2$ is strictly positive, we use the asymptotically linear representation of $\hat{\boldsymbol{p}}^\pi$ owing to the conditional functional $\delta$-method in Theorem~2 of \cite{dobler23}:
to this end, we only consider the first component of $\hat{\boldsymbol{p}}^\pi$ instead of the whole vector. 
Thus, using the notation $\dot \psi_{\tilde {\mathbb P},1}$ for the first component of $\dot \psi_{\tilde {\mathbb P}}$,
the asymptotic conditional variance of $\hat p_1^\pi$, multiplied by $n$, equals
\begin{align*}
    & \frac1n \sum_{k=1}^n 
    \operatorname{Var}( \dot \psi_{\tilde {\mathbb P}}( \Pi_k(\boldsymbol{X}_k, \boldsymbol{\lambda}_k )) \ | \ \boldsymbol{X}_k, \boldsymbol{\lambda}_k ) \\
    & = \frac1n \sum_{k=1}^n \big[
    E( \dot \psi^2_{\tilde {\mathbb P}}( \Pi_k(\boldsymbol{X}_k, \boldsymbol{\lambda}_k )) \ | \ \boldsymbol{X}_k, \boldsymbol{\lambda}_k ) - E( \dot \psi_{\tilde {\mathbb P}}( \Pi_k(\boldsymbol{X}_k, \boldsymbol{\lambda}_k )) \ | \ \boldsymbol{X}_k, \boldsymbol{\lambda}_k )^2\big] \\
    & = \frac1n \sum_{k=1}^n \Big[
    \frac1{|\mathcal{S}|}  \sum_{\pi \in \mathcal{S}}  \dot \psi^2_{\tilde {\mathbb P}}(\boldsymbol{X}_k^\pi, \boldsymbol{\lambda}_k^\pi )  - \Big( \frac1{|\mathcal{S}|}  \sum_{\pi \in \mathcal{S}}  \dot \psi_{\tilde {\mathbb P}}(\boldsymbol{X}_k^\pi, \boldsymbol{\lambda}_k^\pi) \Big)^2\Big] \\
    & \stackrel{\text{a.s.}}{\longrightarrow} E(\operatorname{Var}( \dot \psi_{\tilde {\mathbb P}}( \Pi_1(\boldsymbol{X}_1, \boldsymbol{\lambda}_1 )) \ | \ \boldsymbol{X}_1, \boldsymbol{\lambda}_1 )) =: \tilde \sigma^2
\end{align*}
as $n\to\infty$ by the strong law of large numbers.
Here we have used the notation $\mathcal{S}$ for the employed symmetric group.

It follows that $\tilde \sigma^2 >0$ if $\operatorname{Var}( \dot \psi_{\tilde {\mathbb P},1}( \Pi_1(\boldsymbol{X}_1, \boldsymbol{\lambda}_1 )) \ | \ \boldsymbol{X}_1, \boldsymbol{\lambda}_1 )$ is positive with a non-zero probability.
This is the case when $\dot \psi_{\tilde {\mathbb P}}$ is \textit{not} $\tilde {\mathbb P}$-a.s.\ constant.
To investigate this case in more detail, we first note that $E(\dot\psi_{\RP} (\Pi_1(\boldsymbol{X}_1, \boldsymbol{\lambda}_1)) \ | \ \vX_1, \vlam_1) = \boldsymbol{0}_d$ a.s.\ which easily follows from the structure in the first asymptotically linear representation in~\eqref{eq:asy_lin}.
Similarly, it is noticed that $E(\dot\psi_{\P} (\boldsymbol{X}_1, \boldsymbol{\lambda}_1) ) = \boldsymbol{0}_d$.
Thus, it remains to analyze the (conditional) second moments.
Note that one of the terms contributing to the following sum stems from the identity permutation:
$$\tilde \sigma^2 = E( \dot \psi^2_{\tilde {\mathbb P},1}( \Pi_1(\boldsymbol{X}_1, \boldsymbol{\lambda}_1 )) ) = \frac1{d!} \sum_{\pi \in \mathcal{S}_d} E( \dot \psi^2_{\tilde {\mathbb P},1}( \boldsymbol{X}_1^\pi, \boldsymbol{\lambda}_1^\pi ) ) \geq \frac1{d!} E( \dot \psi^2_{\tilde {\mathbb P},1}( \boldsymbol{X}_1, \boldsymbol{\lambda}_1 ) ). $$
Additionally, since $\RP$ dominates $\P$, it follows that $\sigma^2 = E( \dot \psi^2_{ {\mathbb P},1}( \boldsymbol{X}_1, \boldsymbol{\lambda}_1 ) ) > 0 $ implies  $\tilde \sigma^2 > 0$.
\hfill\qed

\bigskip

\noindent\textbf{Proof of Corollary~\ref{cor:test1}.}
The asymptotic statement is a direct consequence of Theorem~\ref{thm:T} and Corollary~\ref{cor:T_pi}, in view of general theory for resampling-based tests; cf.\ Lemma~1 in \cite{janssen03}.
That a quasi-randomization test controls the type-I error probability under quasi-randomization-invariant sampling distributions is well-known fact; also see \cite{dobler23}.
In the present case though, we should specifically analyze the cases where $T_n$ or $T_n^\pi$ are non-computable due to no observations in at least one component.
We denote these events by $N_n$ and $N_n^\pi$, respectively.
Keeping in mind that $T_n = 0$ and $T_n^\pi = \infty$ under $N_n$ and $N_n^\pi$, respectively, these choices lead to a test that is potentially more conservative than a quasi-randomization test under invariance which is always computable, hence $E(\Psi_n) \leq \alpha$.
Lastly, it should be noted that $P(N_n) \to 0$ and $P(N_n^\pi) \to 0$ exponentially fast:
$$ P(N_n) = P\Big( \bigcup_{i=1}^d \bigcap_{k=1}^n \{ \lambda_{ik} = 0 \} \Big) \leq \sum_{i=1}^d P(  \lambda_{i1} = 0 )^n $$
and, under permutation-invariance of the sampling distribution,
$$ P(N_n^\pi ) = \frac{1}{(d!)^n}\sum_{(\pi^{(1)}, \dots, \pi^{(n)}) \in \mathcal{S}_d^n} P\Big( \bigcup_{i=1}^d \bigcap_{k=1}^n \{ \lambda_{i\pi^{(k)}(k)} = 0 \} \Big) = P\Big( \bigcup_{i=1}^d \bigcap_{k=1}^n \{ \lambda_{ik} = 0 \} \Big) = P(N_n). \hfill\qed$$

\subsection{Matrix ranks in the two-way layout after quasi-randomization}

\subsubsection{Testing for main effects}
\label{ssec:rank_main}

Let us start with analyzing the rank of $\vC_A \tilde{\vV} \vC_A^\top$ where $\tilde \vV = \boldsymbol{I}_{d_A} \otimes (\tilde {\boldsymbol{\Sigma}} - \tilde{\boldsymbol{\Upsilon}}) + \vJ_{d_A} \otimes\tilde{\boldsymbol{\Upsilon}} $. 
By the usual rules for matrix products of Kronecker products,
\begin{align*}
 \vC_A \tilde{\vV} \vC_A^\top & = (\vP_{d_A} \otimes \boldsymbol{1}_{d_B}^\top ) (\boldsymbol{I}_{d_A} \otimes (\tilde {\boldsymbol{\Sigma}} - \tilde{\boldsymbol{\Upsilon}}) + \vJ_{d_A} \otimes\tilde{\boldsymbol{\Upsilon}} ) (\vP_{d_A} \otimes \boldsymbol{1}_{d_B} ) \\
 & = (\vP_{d_A} \vI_{d_A} \vP_{d_A}) \otimes ( \boldsymbol{1}_{d_B}^\top (\tilde {\boldsymbol{\Sigma}} - \tilde{\boldsymbol{\Upsilon}}) \boldsymbol{1}_{d_B} ) 
 + (\vP_{d_A} \vJ_{d_A} \vP_{d_A}) \otimes ( \boldsymbol{1}_{d_B}^\top \tilde {\boldsymbol{\Upsilon}} \boldsymbol{1}_{d_B} ) \\
 & = ( \boldsymbol{1}_{d_B}^\top (\tilde {\boldsymbol{\Sigma}} - \tilde{\boldsymbol{\Upsilon}}) \boldsymbol{1}_{d_B} ) \vP_{d_A}  
 + \boldsymbol{0}_{d_A}.
\end{align*}
This matrix has rank $d_A-1$ if $\gamma_A := \boldsymbol{1}_{d_B}^\top (\tilde {\boldsymbol{\Sigma}} - \tilde{\boldsymbol{\Upsilon}}) \boldsymbol{1}_{d_B} > 0$.
Respectively denote by $\tilde\vZ_1, \tilde \vZ_2$ the first and second $d_B$ components of the limit $\tilde\vZ$ which results from the variant of Theorem~\ref{thm:main_rand} that is coined to the adjusted quasi-randomization scheme.
Note that
$\tilde {\boldsymbol{\Sigma}} = \operatorname{Cov}(\tilde \vZ_i,\tilde \vZ_i), i=1,2$, and $\tilde {\boldsymbol{\Upsilon}}= \operatorname{Cov}(\tilde \vZ_1,\tilde \vZ_2)$.
Exploiting the bilinearity of the covariance,
$\tilde {\boldsymbol{\Sigma}} - \tilde{\boldsymbol{\Upsilon}} = \tfrac12 \operatorname{Cov}(\tilde \vZ_1 - \tilde \vZ_2)$.

Thus, $\gamma_A = \tfrac12  \operatorname{Var}(\boldsymbol{1}_{d_B}^\top (\tilde \vZ_1-\tilde \vZ_2 )) = 0$ if and only if $\boldsymbol{1}_{d_B}^\top (\tilde \vZ_1- \tilde \vZ_2)=0$ almost surely. Using $\vv:= 2^{-1/2} (1,-1,0,...,0)^\top\otimes \boldsymbol{1}_{d_B}$ this means $\vv^\top \tilde \vZ=0$. 
However, using that $\dot\psi_{\RP} (\Pi_1(\boldsymbol{X}_1,\boldsymbol{\lambda}_1))$ has the same first two moments as $\tilde \vZ$, one can show similarly as in the proof of Theorem~\ref{thm:main_rand} that 
$$E( \vv^\top \dot\psi_{\RP} (\Pi_1(\boldsymbol{X}_1,  \boldsymbol{\lambda}_1)) \ | \ \vX_1, \vlam_1) = 0 = E( \vv^\top \dot\psi_{\RP} (\boldsymbol{X}_1, \boldsymbol{\lambda}_1)) $$
almost surely. 
Thus, we continue as in the proof of Theorem~\ref{thm:main_rand} and argue for the asymptotic variance that
\begin{align*}
   \gamma_A & = E( (\vv^\top\dot \psi_{\tilde {\mathbb P}}( \Pi_1(\boldsymbol{X}_1, \boldsymbol{\lambda}_1 )))^2 ) = \frac1{d_A!} \sum_{\pi \in \mathcal{S}_{d_A}} E(( \vv^\top \dot \psi_{\tilde {\mathbb P}}( \boldsymbol{X}_1^\pi, \boldsymbol{\lambda}_1^\pi ))^2 ) \\
   &\geq \frac1{d_A!} E((\vv^\top \dot \psi_{\tilde {\mathbb P}}( \boldsymbol{X}_1, \boldsymbol{\lambda}_1 ))^2 ) = \frac1{d_A!} \vv^\top \vV \vv > 0. 
\end{align*}
This final conclusion is due to the fact that $\vv \notin \ker(\vV) = \{c \boldsymbol{1}_{d_Ad_B}: c \in \R\} $.

\subsubsection{Testing for interaction effects}
\label{ssec:rank_interaction}

We move on to computing the rank of $\vC_{AB} \tilde{\vV} \vC_{AB}^\top$ where $\tilde \vV = \boldsymbol{I}_{d_A} \otimes (\tilde {\boldsymbol{\Sigma}} - \tilde{\boldsymbol{\Upsilon}}) + \vJ_{d_A} \otimes\tilde{\boldsymbol{\Upsilon}} $. 
Although we use the same notation for $\tilde \vV$ as above, we wish to point out that the definitions of the block matrices have changed as a consequence of the changed quasi-randomization scheme:
$$\tilde{\boldsymbol{\Sigma}} = (\tilde \sigma^2 - \tilde \rho) \boldsymbol{I}_{d_B} + \tilde \rho \boldsymbol{J}_{d_B},\qquad \tilde{\boldsymbol{\Upsilon}} = (\tilde \eta - \tilde \nu)  \boldsymbol{I}_{d_B} 
 + \tilde{\nu} \boldsymbol{J}_{d_B}.$$
Using similar arguments as above, 
\begin{align*}
 \vC_{AB} \tilde{\vV} \vC_{AB}^\top & = (\vP_{d_A} \otimes \vP_{d_B} ) \Big( (\tilde \sigma^2 - \tilde \rho) \boldsymbol{I}_{d_A} \otimes \boldsymbol{I}_{d_B} + \tilde \rho \boldsymbol{I}_{d_A} \otimes \boldsymbol{J}_{d_B} \\
 & \quad - (\tilde \eta - \tilde \nu) \boldsymbol{I}_{d_A} \otimes \boldsymbol{I}_{d_B} 
 - \tilde{\nu} \boldsymbol{I}_{d_A} \otimes \boldsymbol{J}_{d_B}
 \\
 & \qquad 
 + (\tilde \eta - \tilde \nu) \vJ_{d_A} \otimes \vI_{d_B} + \tilde \nu \vJ_{d_A} \otimes \vJ_{d_B} \Big) (\vP_{d_A} \otimes \vP_{d_B} )
 \\
 & = (\tilde \sigma^2 - \tilde \rho - \tilde \eta + \tilde \nu)  (\vP_{d_A} \otimes \vP_{d_B} )
 + \boldsymbol{0}_{d_A \times d_B}.
\end{align*}
This matrix has rank $(d_A-1)(d_B-1)$ if $\gamma_{AB} := \tilde \sigma^2 - \tilde \rho - \tilde \eta + \tilde \nu\neq 0$.
Writing the limit random vector as $\tilde \vZ=(\tilde Z_{11}, \tilde Z_{12}, \dots, \tilde Z_{d_A d_B})^\top$, ...
Note that $\gamma_{AB}$ equals
\begin{align*}
    \tfrac14\big[ & \big( \operatorname{Var}(\tilde Z_{11}) + \operatorname{Var}(\tilde Z_{12}) + \operatorname{Var}(\tilde Z_{21}) + \operatorname{Var}(\tilde Z_{22}) \big) + 4 \operatorname{Cov}(\tilde Z_{11}, \tilde Z_{22}) \\
    & - \big( \operatorname{Cov}(\tilde Z_{11}, \tilde Z_{12}) + \operatorname{Cov}(\tilde Z_{12}, \tilde Z_{11}) + \operatorname{Cov}(\tilde Z_{21}, \tilde Z_{22}) + \operatorname{Cov}(\tilde Z_{22}, \tilde Z_{21}) \big) \\
    & - \big( \operatorname{Cov}(\tilde Z_{11}, \tilde Z_{21}) + \operatorname{Cov}(\tilde Z_{21}, \tilde Z_{11}) + \operatorname{Cov}(\tilde Z_{12}, \tilde Z_{22}) + \operatorname{Cov}(\tilde Z_{22}, \tilde Z_{12}) \big) \\
   = & \tfrac14  \operatorname{Var}(\tilde Z_{11} - \tilde Z_{21} - \tilde Z_{12} + \tilde Z_{22})
\end{align*}
Thus, for $\vv := \tfrac12 (1, -1)^\top \oplus \boldsymbol{1}_{d_B-2} \oplus (-1, 1)^\top \oplus \boldsymbol{1}_{(d_A-2) d_B + (d_B-2)}  $,  we have $\gamma_{AB} =  \operatorname{Var}(\vv^\top \tilde \vZ) = \vv^\top \tilde {\boldsymbol{V}} \vv > 0$ for similar reasons as in the previous subsection.

\section{Asymptotic variance-covariance matrix and estimators}
\label{app:covariance}

In the following, we compute the limiting covariance of the components $i,j \in \{1,\dots, d\}$ of the normalized Mann-Whitney effect estimator:
 
\begin{align*}
\sigma_{ij} & := \operatorname{Cov}(\dot\psi_{i,\mathbb{P}} (\boldsymbol{X}_1, \boldsymbol{\lambda}_1), \dot\psi_{j,\mathbb{P}} (\boldsymbol{X}_1, \boldsymbol{\lambda}_1)) \\
& = \operatorname{Cov}(\dot\psi_{i2,(F_1, \dots, F_d)} \cdot (\dot\psi_{11,(\check F_1, \kappa_1)}\cdot (\lambda_{11} \cdot c(\ \cdot \ - X_{11}), \lambda_{11}), \\ 
 & \hspace{3.2cm} \dots,  \dot\psi_{d1,(\check F_d, \kappa_d)} \cdot (\lambda_{d1} \cdot c(\ \cdot \ - X_{d1}), \lambda_{d1})) , \\
& \qquad \dot\psi_{j2,(F_1, \dots, F_d)} \cdot (\dot\psi_{11,(\check F_1, \kappa_1)}\cdot (\lambda_{11} \cdot c(\ \cdot \ - X_{11}), \lambda_{11}), \\ 
 & \hspace{3.2cm} \dots,  \dot\psi_{d1,(\check F_d, \kappa_d)} \cdot (\lambda_{d1} \cdot c(\ \cdot \ - X_{d1}), \lambda_{d1}))) \\
& = cov\Big(\frac1d\sum_{s=1}^d  \Big[ \int \Big\{\frac{\lambda_{s1} \cdot c( \ \cdot \ - X_{s1})}{\kappa_s} - \frac{\lambda_{s1} \cdot \check F_s}{\kappa_s^2} \Big\}d F_i 
- \int \Big\{\frac{\lambda_{i1} \cdot c( \ \cdot \ - X_{i1})}{\kappa_i} - \frac{\lambda_{i1}\cdot\check F_i}{\kappa_i^2} \Big\}d F_s \Big], \\
& \qquad \frac1d\sum_{t=1}^d  \Big[ \int \Big\{\frac{\lambda_{t1} \cdot c( \ \cdot \ - X_{t1})}{\kappa_t} - \frac{\lambda_{t1}\cdot\check F_t}{\kappa_t^2} \Big\}d F_j 
- \int \Big\{\frac{\lambda_{j1} \cdot c( \ \cdot \ - X_{j1})}{\kappa_j} - \frac{\lambda_{j1}\cdot\check F_j}{\kappa_j^2} \Big\}d F_{t} \Big]
\Big)
\end{align*}
Note that the terms in square brackets are centered.
In addition, we now use that 
$$\check F_s(x) = P(X_{s1} < x, \lambda_{s1} = 1) + \tfrac12 P(X_{s1} = x, \lambda_{s1} = 1) = \kappa_s F_s(x), $$
owing to the MCAR assumption.
Hence, the asymptotic covariance expression simplifies to
\begin{align*}
    \frac1{d^2} \sum_{s,t =1}^d E\Big( & \Big[\frac{\lambda_{s1}}{\kappa_s}\Big( \int c( \ \cdot \ - X_{s1}) d F_i - p^{(si)} \Big) - 
    \frac{\lambda_{i1}}{\kappa_i}\Big( \int c( \ \cdot \ - X_{i1}) d F_s - p^{(is)} \Big)     \Big] \\
     \times & \Big[\frac{\lambda_{t1}}{\kappa_t}\Big( \int c( \ \cdot \ - X_{t1}) d F_j - p^{(tj)} \Big) - 
    \frac{\lambda_{j1}}{\kappa_j}\Big( \int c( \ \cdot \ - X_{j1}) d F_t - p^{(jt)} \Big)     \Big]\Big).
\end{align*}
Here, $p^{(ij)} = \int F_i d F_j$.
Using the definitions $\kappa_{ij}=P(\lambda_{i1}=1, \lambda_{j1}=1)$ and 
$$F_{ij}(x,y) = \tfrac14(P(X_{i1} < x, X_{j1}< y) + P(X_{i1}  \leq  x, X_{j1}< y) $$
$$ \qquad \qquad \qquad + \ P(X_{i1} < x, X_{j1} \leq y) + P(X_{i1} \leq  x, X_{j1} \leq y) ) $$
and the independence of the $\lambda$'s and $X$'s,
the above covariance equals
\begin{align*}
    \frac1{d^2} \sum_{s,t =1}^d 
    & \Big[ \frac{\kappa_{st}}{\kappa_s \kappa_t} \Big(\int \int F_{st}(x,y)dF_i(x) d F_j(y) - p^{(si)} p^{(tj)} \Big) \\
    & - \frac{\kappa_{it}}{\kappa_i \kappa_t} \Big(\int \int F_{it}(x,y)dF_s(x) d F_j(y) - p^{(is)} p^{(tj)} \Big) \\
    & - \frac{\kappa_{sj}}{\kappa_s \kappa_j} \Big(\int \int F_{sj}(x,y)dF_i(x) d F_t(y) - p^{(si)} p^{(jt)} \Big) \\
    & + \frac{\kappa_{ij}}{\kappa_i \kappa_j} \Big(\int \int F_{ij}(x,y)dF_s(x) d F_t(y) - p^{(is)} p^{(jt)} \Big)
    \Big] \\
    = \frac1{d^2} \sum_{s,t =1}^d 
    & \Big[ \frac{\kappa_{st}}{\kappa_s \kappa_t} (p^{(si;tj)} - p^{(si)} p^{(tj)} ) 
    - \frac{\kappa_{it}}{\kappa_i \kappa_t} (p^{(is;tj)} - p^{(is)} p^{(tj)} ) \\
     & - \frac{\kappa_{sj}}{\kappa_s \kappa_j} (p^{(si;jt)} - p^{(si)} p^{(jt)} )
    + \frac{\kappa_{ij}}{\kappa_i \kappa_j} (p^{(is;jt)} - p^{(is)} p^{(jt)} )
    \Big]
\end{align*}
Here, $p^{(si;tj)} := \int \int F_{st}(x,y)dF_i(x) d F_j(y) $.
Note that $\kappa_{ss}= \kappa_s$  and  $p^{(ss)}=\tfrac12$
but, when $x \leq y$, 
$$F_{ss}(x,y) = \ind\{x=y\} (\tfrac34 P(X_{s1} < x) + \tfrac14 P(X_{s1} \leq x))  + \ind\{x < y\} (\tfrac12 P(X_{s1} < x) + \tfrac12 P(X_{s1} \leq x))  $$ 
 differs from $F_s(\min(x,y))$, so further simplifications 
are not very fruitful.

Replacing the quantities in the above expression by their empirical counterparts yields a plug-in estimator of $\sigma_{ij}$. This estimator is consistent, since it is a linear combination of consistent estimators. Although this provides a consistent covariance estimator when all entries are estimated componentwise, it is more convenient and structurally preferable to estimate the covariance matrix as a whole. In particular, this allows one to preserve structural properties of covariance matrices, such as the positive semi-definiteness of the resulting estimator.
Thus, an estimator for the asymptotic variance-covariance matrix can be directly based on the vector-valued asymptotic linearity of the estimator, as stated in the proof of Theorem~\ref{thm:clt} at the beginning of Appendix~\ref{app:proofs}:
the estimated influence function of the $k$-th random vector is\begin{align*}
    \hat{\dot{\boldsymbol{\psi}}}_k =  \frac nd\sum_{s=1}^d &  \Big( \Big[  \int \Big\{ \frac{(\ind\{X_{sk} < x\} + \tfrac12 \ind\{X_{sk} = x\}) \lambda_{sk}}{\lambda_{s\cdot}} - \frac{\hat{\check F}_s(x) \lambda_{sk}}{\lambda_{s\cdot}^2/n} \Big\} d \hat F_i(x) \\
    & - \int\Big\{ \frac{(\ind\{X_{ik} < x\} + \tfrac12 \ind\{X_{ik} = x\}) \lambda_{ik}}{\lambda_{i\cdot}} - \frac{\hat{\check F}_i(x) \lambda_{ik}}{\lambda_{i\cdot}^2/n} \Big\} d \hat  F_s(x) \Big] : i = 1 ,\dots, d \Big)^\top.
\end{align*} 
Where $\hat{\check F}_i(x) = \frac1n \sum_{\ell=1}^n (\ind\{X_{i\ell} < x \} + \tfrac12 \ind\{X_{i\ell} = x\}) \lambda_{i\ell}$.
Note that the average of these terms over $k=1,\dots, n$ reduces to zero.
Hence, we propose the following easy-to-implement estimator of the asymptotic variance-covariance matrix:
\begin{align*}
    \hat{\boldsymbol{V}}:=
     \frac1{{n}} \sum_{k=1}^n  \hat{\dot{\boldsymbol{\psi}}}_k\hat{\dot{\boldsymbol{\psi}}}_k^\top.
\end{align*}
This representation guarantees positive semi-definiteness, a desirable property, which  is not generally ensured in alternative estimators such as standard plug-in estimators.

We derive the following alternative representation of the estimated influence function which facilitates software implementations of $\hat{\vV}$:
\begin{align*}
     \hat{\dot{\boldsymbol{\psi}}}_k &
     = \frac nd\sum_{s=1}^d \sum_{\ell=1}^n   \Big( \Big[ \frac{\lambda_{sk}\lambda_{i\ell}}{\lambda_{s\cdot}\lambda_{i\cdot}}  \Big\{ (\ind\{X_{sk} < X_{i \ell}\} + \tfrac12 \ind\{X_{sk} = X_{i\ell}\})  - \frac{\hat{\check F}_s(X_{i\ell})}{\lambda_{s\cdot}/n} \Big\}  \\
    & - \frac{\lambda_{ik}\lambda_{s\ell}}{\lambda_{i\cdot}\lambda_{s\cdot}} \Big\{ (\ind\{X_{ik} < X_{s \ell}\} + \tfrac12 \ind\{X_{ik} = X_{s\ell}\})  - \frac{\hat{\check F}_i(X_{s \ell}) }{\lambda_{i\cdot}/n} \Big\} \Big] : i = 1 ,\dots, d \Big)^\top \\
    & = \frac1 d\sum_{s=1}^d \sum_{j,\ell=1}^n   \Big( \Big[ \frac{\lambda_{sk}\lambda_{i\ell}}{\lambda_{s\cdot}\lambda_{i\cdot}}  \Big\{ (\ind\{X_{sk} < X_{i \ell}\} + \tfrac12 \ind\{X_{sk} = X_{i\ell}\}) \\
    & \qquad \qquad \qquad \qquad \qquad \qquad - (\ind\{X_{sj} < X_{i \ell}\} + \tfrac12 \ind\{X_{sj} = X_{i\ell}\}) \frac{n\lambda_{sj}}{\lambda_{s\cdot}} \Big\}  \\
    &  \qquad \qquad \qquad  
    - \frac{\lambda_{ik}\lambda_{s\ell}}{\lambda_{i\cdot}\lambda_{s\cdot}} \Big\{ (\ind\{X_{ik} < X_{s \ell}\} + \tfrac12 \ind\{X_{ik} = X_{s\ell}\}) \\
    & \qquad \qquad \qquad \qquad \qquad \qquad
    - (\ind\{X_{ij} < X_{s\ell}\} + \tfrac12 \ind\{X_{ij} = X_{s\ell}\} ) \frac{n \lambda_{ij}}{\lambda_{i\cdot}} \Big\} \Big] : i = 1 ,\dots, d \Big)^\top 
\end{align*}

\begin{thm}
\label{thm:cov_consistency}
    Under Assumptions~\ref{ass:MCAR}--\ref{ass:test}, the covariance estimators $\hat \sigma_{ij}, i,j \in \{1,\dots, d\}$, and $\hat{\boldsymbol{V}}$ are strongly consistent for $\sigma_{ij}$ and $\boldsymbol{V}$, respectively.
    Additionally, $P(\mathrm{rk}(\hat{\boldsymbol{V}}) = d-1) \to 1$ as $n\to \infty$.
\end{thm}
\noindent 
\textbf{Proof of Theorem~\ref{thm:cov_consistency}.}
    It is obvious that the estimators $\hat \sigma_{ij}$ are continuous functionals of $\hat F_1, \dots, \hat F_d$ and $\hat F_{st}, s,t \in \{1,\dots,d\}$.
    Owing to the strong law of large numbers, these functions are consistent for $F_1, \dots, F_d$, and $F_{st}, s,t \in \{1,\dots,d\}$, respectively.
    This similarly holds true for each component of $\hat{\vV}$: 
    the derivations above reveal that the components of the (almost sure) limit $\hat{\vV}$ must be equal to $\sigma_{ij}, i,j\in \{1,\dots,d\}$.
    Finally, due to the lower semi-continuity of the rank mapping, combined with the assumption that $\mathrm{rk}(\vV) = d-1$, the claimed convergence of $\mathrm{rk}(\hat{\boldsymbol{V}})$ follows. 
    It actually even holds almost surely.
    \hfill \qed

\section{An Example with Low-Rank Data and Full Contrast Rank}
\label{app:rankexample}

For the sake of simplicity consider the complete-data case, i.e., $\lambda_{i\cdot}=n$ for all $i=1,\dots,d$, and let $ \boldsymbol X_k = (X_{1k},X_{2k},X_{3k},X_{4k})^\top,$ $ k=1,\dots,n, $ be i.i.d. four-dimensional random vectors with\[ (X_{1k},X_{2k})^\top\sim \mathcal N_2(\boldsymbol{0}_2,\vI_2), \qquad X_{3k}\equiv0, \qquad X_{4k}\equiv1. \] 
Hence, only the first two components are non-deterministic, resulting in the covariance matrix  $ \operatorname{Cov}(\boldsymbol X_k) = \operatorname{diag}(1,1,0,0), $ which has rank two. However, as we will show below, it holds for the covariance matrix $\vV$ related to the Mann-Whitney effect estimator that $\mathrm{rk}(\vV)=3=d-1$. 
With the definitions $a:=\Phi(1)$ and $q:=1-a$, we obtain the following matrix of pair-wise Mann-Whitney effects:
$$(p^{(\ell i)})_{\ell,i=1}^4 = (P(X_{\ell1} < X_{i1}) +\tfrac12 P(X_{\ell1} < X_{i1}))_{\ell,i=1}^4 
=\frac 1 2 \begin{pmatrix} 1 & 1 & 1 & 2a\\ 1 &1 & 1 & 2a\\ 1 & 1 &1 & 2\\ 2q & 2q & 0 & 1 \end{pmatrix}. $$
Now, by using the Hoeffding decomposition for the corresponding two-sample \(U\)-statistics \citep{Hoeffding1948}, we obtain the first-order representation 
$$\sqrt n(\hat{\vp}-\vp) = \frac14\vL\vZ+\lan_p(1), $$ with 
$$ \vL= \begin{pmatrix} 1&1&1&-1&0&0\\ -1&0&0&1&1&1\\ 0&-1&0&0&-1&0\\ 0&0&-1&0&0&-1 \end{pmatrix}, $$ and 
$$ \vZ = \begin{pmatrix} Z_{A,1}\\ Z_{B,1}\\ Z_{C,1}\\ Z_{A,2}\\ Z_{B,2}\\ Z_{C,2} \end{pmatrix} \overset{d}{\longrightarrow} \mathcal N_6 \left( \vnull_6, \begin{pmatrix} \vM&\vnull_{3\times3}\\ \vnull_{3\times3}&\vM \end{pmatrix} \right) $$ with 
$$ \vM = \begin{pmatrix} \frac1{12} & \frac18 & \frac{aq}{2}\\[1mm] \frac18 & \frac14 & \frac q2\\[1mm] \frac{aq}{2} & \frac q2 & aq \end{pmatrix}. $$
This results in
$$\sqrt n(\widehat{\vp}-\vp) \overset{d}{\longrightarrow} \mathcal N_4(\boldsymbol 0_4,\vV), \qquad \text{where} \quad  \vV = \frac1{16} \vL \begin{pmatrix} \vM&\vnull_{3\times3}\\ \vnull_{3\times3}&\vM \end{pmatrix} \vL^\top. $$

Annother way to derive the covariance matrix $\vV$ would be applying the general formula for $\sigma_{ij}$ from Appendix~\ref{app:covariance}.
 This approach would avoid the Hoeffding decomposition, but it would entail a lengthy series of elementary case-by-case computations. For this reason, we presented the shorter Hoeffding-based derivation above, which yields the same result.

 A direct calculation of the upper-left \(3\times3\) principal minor of $\vV$ yields
\[
\det(\vV_{1:3,1:3})
=
\frac{(a-1)(2a-1)(36a^2-24a-25)}{98304}.
\]
which is positive for \(a=\Phi(1)\). 
Consequently, we conclude that
\(\mathrm{rk}(\vV)\geq 3\), implying
$
\mathrm{rk}(\vV)=3=d-1,
$
although the covariance matrix of the underlying random vectors $\vX_k$ has only rank two. This
illustrates that the condition \(\mathrm{rk}(\vV)=d-1\) can be substantially less
restrictive than a corresponding  assumption on the covariance matrix
of the underlying data.

\FloatBarrier

\section{Detailed simulation results under MCAR}\label{app:add_sim}
Tables~\ref{tab:simulation_mcar_different_11x1_tests_rows}--\ref{tab:simulation_mcar_different_3x2_ab_tests_rows} contain the empirical rejection rates for our considered settings; for $\delta=0$, i.e., under the respective null hypotheses, type-I error proportions within the binomial $95\%$-interval $[4.42\%,5.62\%]$ based on 5,000 level-$5\%$-tests are printed in bold.
In particular, we see that all but two values of the simulated sizes of $\Psi_n^\pi$ are contained in this interval.
For $\Psi_n^{\pi\min}$, which is only available for testing the main and interaction effects in the two-way layout, only one of the simulated sizes are not contained in the interval.
In comparison, nine of the simulated sizes of $\Psi_n^\textnormal{boot}$ are not contained; this test seems generally somewhat more liberal than the proposed tests.
Lastly, the asymptotic test $\Psi_n^\textnormal{asy}$ is extremely liberal, with simulated sizes ranging from about 6\% to about 58\%.

\begin{table}[h]
\caption{Empirical rejection rates in percent for the one-factor setting with $d=11$ under the MCAR missingness mechanism with component-specific missingness probabilities. The hypothesis $H_{0,\textnormal{equal}}$ is tested at the nominal significance level $\alpha=5\%$. The case $\delta=0$ corresponds to the null hypothesis, whereas $\delta>0$ represents a trend alternative.}
\label{tab:simulation_mcar_different_11x1_tests_rows}
\centering
\small
\setlength{\tabcolsep}{3.5pt}
\renewcommand{\arraystretch}{1.05}

\begin{tabular}{@{}lll*{9}{r}@{}}
\hline
&&&\multicolumn{3}{c}{$n=20$} & \multicolumn{3}{c}{$n=35$} & \multicolumn{3}{c}{$n=50$} \\
\cline{4-6}
\cline{7-9}
\cline{10-12}
Copula & $\rho_X$ & Test & \multicolumn{3}{c}{$\delta$} & \multicolumn{3}{c}{$\delta$} & \multicolumn{3}{c}{$\delta$} \\
\cline{4-6}
\cline{7-9}
\cline{10-12}
&&& 0.00 & 0.15 & 0.30 & 0.00 & 0.15 & 0.30 & 0.00 & 0.15 & 0.30 \\
\hline

Clayton & $0.1$ & $\Psi^\pi_n$ & \textbf{5.60} & 8.50 & 21.40 & \textbf{5.06} & 12.86 & 50.34 & \textbf{5.24} & 19.44 & 72.90 \\
 &  & $\Psi^\textnormal{boot}_n$ & \textbf{5.58} & 8.84 & 21.88 & \textbf{4.96} & 13.06 & 50.26 & \textbf{5.24} & 19.46 & 72.76 \\
 &  & $\Psi^\textnormal{asy}_n$ & 57.62 & 66.34 & 86.66 & 29.70 & 48.02 & 86.16 & 19.90 & 46.00 & 92.02 \\[4pt]
 & $0.2$ & $\Psi^\pi_n$ & \textbf{4.98} & 8.08 & 21.18 & \textbf{5.22} & 13.64 & 49.78 & \textbf{5.10} & 19.08 & 72.88 \\
 &  & $\Psi^\textnormal{boot}_n$ & \textbf{5.52} & 8.72 & 22.52 & \textbf{5.52} & 14.26 & 50.92 & \textbf{5.26} & 19.74 & 73.58 \\
 &  & $\Psi^\textnormal{asy}_n$ & 54.94 & 65.28 & 84.50 & 28.20 & 47.24 & 84.48 & 18.76 & 44.44 & 90.78 \\
\hline

Gumbel & $0.1$ & $\Psi^\pi_n$ & \textbf{5.28} & 10.14 & 26.06 & \textbf{5.46} & 16.24 & 54.00 & \textbf{4.96} & 20.98 & 77.54 \\
 &  & $\Psi^\textnormal{boot}_n$ & \textbf{5.48} & 10.44 & 25.84 & \textbf{5.60} & 16.08 & 53.80 & \textbf{4.88} & 20.96 & 77.18 \\
 &  & $\Psi^\textnormal{asy}_n$ & 56.52 & 68.04 & 87.30 & 30.32 & 52.32 & 87.80 & 20.28 & 47.80 & 93.44 \\[4pt]
 & $0.2$ & $\Psi^\pi_n$ & \textbf{5.58} & 12.42 & 29.18 & 5.68 & 17.96 & 59.66 & \textbf{5.20} & 24.94 & 79.64 \\
 &  & $\Psi^\textnormal{boot}_n$ & 5.98 & 12.84 & 30.18 & 5.96 & 18.44 & 59.96 & \textbf{5.24} & 25.38 & 79.72 \\
 &  & $\Psi^\textnormal{asy}_n$ & 53.86 & 68.12 & 87.68 & 27.74 & 52.76 & 88.56 & 18.54 & 50.24 & 93.80 \\
\hline
\end{tabular}

\end{table}
\begin{table}[p]
\caption{Empirical rejection rates in percent for the two-way layout  $d_A = 3$ levels for factor $A$ and $d_B = 2$ levels for factor $B$, (resulting in $d=d_Ad_B=6$) under the MCAR missingness mechanism with component-specific missingness probabilities. The hypothesis $H_{0,\textnormal{equal}}$ is tested at the nominal significance level $\alpha=5\%$. The case $\delta=0$ corresponds to the null hypothesis, whereas $\delta>0$ represents a one-component alternative.}
\label{tab:simulation_mcar_different_3x2_equal_tests_rows}
\centering
\small
\setlength{\tabcolsep}{3.5pt}
\renewcommand{\arraystretch}{1.05}

\begin{tabular}{@{}lll*{9}{r}@{}}
\hline
&&&\multicolumn{3}{c}{$n=20$} & \multicolumn{3}{c}{$n=35$} & \multicolumn{3}{c}{$n=50$} \\
\cline{4-6}
\cline{7-9}
\cline{10-12}
Copula & $\rho_X$ & Test & \multicolumn{3}{c}{$\delta$} & \multicolumn{3}{c}{$\delta$} & \multicolumn{3}{c}{$\delta$} \\
\cline{4-6}
\cline{7-9}
\cline{10-12}
&&& 0.00 & 0.15 & 0.30 & 0.00 & 0.15 & 0.30 & 0.00 & 0.15 & 0.30 \\
\hline

Clayton & $0.1$ & $\Psi^\pi_n$ & \textbf{5.14} & 21.30 & 72.06 & \textbf{5.58} & 40.68 & 96.92 & \textbf{4.66} & 59.56 & 99.70 \\
 &  & $\Psi^\textnormal{boot}_n$ & 5.74 & 23.04 & 73.92 & 5.72 & 41.30 & 97.18 & \textbf{4.90} & 59.80 & 99.72 \\
 &  & $\Psi^\textnormal{asy}_n$ & 21.86 & 52.82 & 93.46 & 13.60 & 61.26 & 99.22 & 9.66 & 72.76 & 99.98 \\[4pt]
 & $0.2$ & $\Psi^\pi_n$ & \textbf{5.48} & 26.02 & 81.70 & \textbf{5.38} & 48.60 & 98.66 & \textbf{5.26} & 65.84 & 100.00 \\
 &  & $\Psi^\textnormal{boot}_n$ & 6.28 & 28.34 & 83.84 & 5.74 & 49.98 & 98.82 & \textbf{5.62} & 66.78 & 100.00 \\
 &  & $\Psi^\textnormal{asy}_n$ & 20.78 & 55.52 & 96.26 & 12.40 & 66.62 & 99.60 & 9.42 & 77.04 & 100.00 \\
\hline

Gumbel & $0.1$ & $\Psi^\pi_n$ & \textbf{4.56} & 20.24 & 72.96 & \textbf{5.32} & 39.36 & 97.40 & \textbf{5.10} & 57.42 & 99.92 \\
 &  & $\Psi^\textnormal{boot}_n$ & \textbf{4.96} & 21.90 & 74.86 & \textbf{5.44} & 39.86 & 97.34 & \textbf{5.12} & 58.26 & 99.94 \\
 &  & $\Psi^\textnormal{asy}_n$ & 22.14 & 52.02 & 93.44 & 13.44 & 59.36 & 99.18 & 10.10 & 71.38 & 99.96 \\[4pt]
 & $0.2$ & $\Psi^\pi_n$ & \textbf{4.86} & 23.56 & 81.38 & \textbf{4.82} & 45.98 & 98.94 & \textbf{5.06} & 67.66 & 99.96 \\
 &  & $\Psi^\textnormal{boot}_n$ & 5.66 & 25.22 & 83.00 & \textbf{5.18} & 47.26 & 99.06 & \textbf{5.34} & 68.04 & 99.98 \\
 &  & $\Psi^\textnormal{asy}_n$ & 19.36 & 52.86 & 95.74 & 11.92 & 64.20 & 99.66 & 9.68 & 78.38 & 99.98 \\
\hline
\end{tabular}

\end{table}

\begin{table}[p]
\caption{Empirical rejection rates in percent for the two-way layout with $d_A=3$ levels for factor $A$ and $d_B=2$ levels for factor $B$, (resulting in $d=d_Ad_B=6$) under the MCAR missingness mechanism with component-specific missingness probabilities. The hypothesis $H_{0A}$ is tested at the nominal significance level $\alpha=5\%$. The case $\delta=0$ corresponds to the null hypothesis, whereas $\delta>0$ represents a one-component alternative.}
\label{tab:simulation_mcar_different_3x2_a_tests_rows}
\centering
\small
\setlength{\tabcolsep}{3.5pt}
\renewcommand{\arraystretch}{1.05}

\begin{tabular}{@{}lll*{9}{r}@{}}
\hline
&&&\multicolumn{3}{c}{$n=20$} & \multicolumn{3}{c}{$n=35$} & \multicolumn{3}{c}{$n=50$} \\
\cline{4-6}
\cline{7-9}
\cline{10-12}
Copula & $\rho_X$ & Test & \multicolumn{3}{c}{$\delta$} & \multicolumn{3}{c}{$\delta$} & \multicolumn{3}{c}{$\delta$} \\
\cline{4-6}
\cline{7-9}
\cline{10-12}
&&& 0.00 & 0.15 & 0.30 & 0.00 & 0.15 & 0.30 & 0.00 & 0.15 & 0.30 \\
\hline

Clayton & $0.1$ & $\Psi^\pi_n$ & \textbf{4.58} & 40.58 & 94.68 & \textbf{4.44} & 67.62 & 99.94 & \textbf{5.16} & 83.66 & 100.00 \\
 &  & $\Psi^{\pi\min}_n$ & \textbf{4.50} & 39.80 & 94.40 & \textbf{4.62} & 67.26 & 99.94 & \textbf{5.20} & 83.42 & 100.00 \\
 &  & $\Psi^\textnormal{boot}_n$ & \textbf{4.82} & 41.48 & 94.82 & \textbf{4.72} & 67.80 & 99.92 & \textbf{5.20} & 83.82 & 100.00 \\
 &  & $\Psi^\textnormal{asy}_n$ & 9.34 & 53.86 & 97.38 & 7.16 & 73.82 & 99.94 & 6.64 & 86.32 & 100.00 \\[4pt]
 & $0.2$ & $\Psi^\pi_n$ & \textbf{4.94} & 46.88 & 97.10 & \textbf{4.86} & 73.58 & 99.96 & \textbf{4.64} & 88.84 & 100.00 \\
 &  & $\Psi^{\pi\min}_n$ & \textbf{5.06} & 46.86 & 96.86 & \textbf{4.90} & 73.40 & 99.96 & \textbf{4.58} & 88.76 & 100.00 \\
 &  & $\Psi^\textnormal{boot}_n$ & \textbf{5.34} & 48.26 & 97.24 & \textbf{5.00} & 73.80 & 99.96 & \textbf{4.88} & 89.10 & 100.00 \\
 &  & $\Psi^\textnormal{asy}_n$ & 9.10 & 58.44 & 98.52 & 7.02 & 78.98 & 99.98 & 6.26 & 91.22 & 100.00 \\
\hline

Gumbel & $0.1$ & $\Psi^\pi_n$ & \textbf{4.54} & 41.04 & 95.66 & \textbf{4.68} & 66.96 & 99.82 & 4.24 & 85.16 & 100.00 \\
 &  & $\Psi^{\pi\min}_n$ & \textbf{4.72} & 41.10 & 95.44 & \textbf{4.64} & 66.82 & 99.80 & 4.24 & 84.72 & 100.00 \\
 &  & $\Psi^\textnormal{boot}_n$ & \textbf{4.78} & 41.66 & 95.80 & \textbf{4.64} & 67.36 & 99.82 & 4.38 & 85.04 & 100.00 \\
 &  & $\Psi^\textnormal{asy}_n$ & 8.24 & 54.56 & 98.10 & 6.84 & 73.56 & 99.88 & 5.88 & 87.90 & 100.00 \\[4pt]
 & $0.2$ & $\Psi^\pi_n$ & \textbf{4.92} & 46.58 & 97.66 & \textbf{4.74} & 74.48 & 99.98 & \textbf{4.90} & 88.84 & 100.00 \\
 &  & $\Psi^{\pi\min}_n$ & \textbf{4.96} & 46.36 & 97.52 & \textbf{4.50} & 74.60 & 99.96 & \textbf{4.72} & 88.86 & 100.00 \\
 &  & $\Psi^\textnormal{boot}_n$ & \textbf{5.18} & 47.50 & 97.80 & \textbf{4.86} & 74.74 & 99.96 & \textbf{4.88} & 88.74 & 100.00 \\
 &  & $\Psi^\textnormal{asy}_n$ & 9.04 & 58.46 & 98.94 & 7.10 & 80.44 & 99.98 & 6.54 & 91.18 & 100.00 \\
\hline
\end{tabular}

\end{table}

\begin{table}[p]
\caption{Empirical rejection rates in percent for the two-way layout with $d_A=3$ levels for factor $A$ and $d_B=2$ levels for factor $B$, (resulting in $d=d_Ad_B=6$) under the MCAR missingness mechanism with component-specific missingness probabilities. The hypothesis $H_{0AB}$ is tested at the nominal significance level $\alpha=5\%$. The case $\delta=0$ corresponds to the null hypothesis, whereas $\delta>0$ represents a one-component alternative.}
\label{tab:simulation_mcar_different_3x2_ab_tests_rows}
\centering
\small
\setlength{\tabcolsep}{3.5pt}
\renewcommand{\arraystretch}{1.05}

\begin{tabular}{@{}lll*{9}{r}@{}}
\hline
&&&\multicolumn{3}{c}{$n=20$} & \multicolumn{3}{c}{$n=35$} & \multicolumn{3}{c}{$n=50$} \\
\cline{4-6}
\cline{7-9}
\cline{10-12}
Copula & $\rho_X$ & Test & \multicolumn{3}{c}{$\delta$} & \multicolumn{3}{c}{$\delta$} & \multicolumn{3}{c}{$\delta$} \\
\cline{4-6}
\cline{7-9}
\cline{10-12}
&&& 0.00 & 0.15 & 0.30 & 0.00 & 0.15 & 0.30 & 0.00 & 0.15 & 0.30 \\
\hline

Clayton & $0.1$ & $\Psi^\pi_n$ & \textbf{5.36} & 13.68 & 45.88 & \textbf{4.94} & 24.78 & 74.14 & \textbf{5.02} & 34.06 & 89.42 \\
 &  & $\Psi^{\pi\min}_n$ & \textbf{5.40} & 13.64 & 45.72 & \textbf{4.94} & 25.04 & 73.92 & \textbf{5.00} & 34.18 & 89.58 \\
 &  & $\Psi^\textnormal{boot}_n$ & 5.78 & 14.02 & 46.28 & \textbf{5.16} & 25.28 & 74.36 & \textbf{5.02} & 34.42 & 89.42 \\
 &  & $\Psi^\textnormal{asy}_n$ & 10.16 & 22.86 & 59.14 & 7.52 & 31.36 & 79.76 & 6.72 & 39.38 & 92.00 \\[4pt]
 & $0.2$ & $\Psi^\pi_n$ & \textbf{5.02} & 16.82 & 52.72 & \textbf{5.04} & 27.38 & 81.20 & \textbf{4.98} & 40.30 & 93.38 \\
 &  & $\Psi^{\pi\min}_n$ & \textbf{4.90} & 17.02 & 52.52 & \textbf{5.18} & 27.58 & 81.00 & \textbf{5.04} & 40.00 & 93.50 \\
 &  & $\Psi^\textnormal{boot}_n$ & \textbf{5.38} & 17.56 & 53.80 & \textbf{5.36} & 28.02 & 81.24 & \textbf{5.20} & 40.36 & 93.54 \\
 &  & $\Psi^\textnormal{asy}_n$ & 9.16 & 26.68 & 64.28 & 7.38 & 33.86 & 85.76 & 6.74 & 44.48 & 94.96 \\
\hline

Gumbel & $0.1$ & $\Psi^\pi_n$ & \textbf{4.52} & 14.04 & 46.04 & \textbf{5.18} & 24.32 & 73.66 & \textbf{5.08} & 35.30 & 89.18 \\
 &  & $\Psi^{\pi\min}_n$ & \textbf{4.54} & 13.94 & 46.16 & \textbf{5.06} & 24.54 & 73.88 & \textbf{4.84} & 35.16 & 88.78 \\
 &  & $\Psi^\textnormal{boot}_n$ & \textbf{4.68} & 14.58 & 46.94 & \textbf{5.26} & 24.56 & 73.78 & \textbf{4.98} & 35.56 & 89.04 \\
 &  & $\Psi^\textnormal{asy}_n$ & 9.00 & 22.54 & 59.70 & 7.34 & 30.20 & 79.64 & 6.76 & 39.72 & 91.24 \\[4pt]
 & $0.2$ & $\Psi^\pi_n$ & \textbf{5.14} & 17.04 & 53.34 & \textbf{4.90} & 29.84 & 81.86 & \textbf{5.28} & 40.24 & 93.82 \\
 &  & $\Psi^{\pi\min}_n$ & \textbf{5.00} & 16.82 & 53.50 & \textbf{5.04} & 29.86 & 81.50 & \textbf{5.24} & 40.20 & 93.84 \\
 &  & $\Psi^\textnormal{boot}_n$ & \textbf{5.52} & 17.50 & 54.48 & \textbf{5.10} & 30.12 & 82.16 & \textbf{5.42} & 40.70 & 93.98 \\
 &  & $\Psi^\textnormal{asy}_n$ & 9.10 & 25.18 & 65.22 & 7.16 & 36.10 & 86.38 & 6.42 & 45.30 & 95.34 \\
\hline
\end{tabular}

\end{table}

\FloatBarrier

\section{Additional simulation results under MAR: \newline robustness against violations of the MCAR assumption}
\label{app:MAR}

We investigated the robustness of the developed tests against violations of the MCAR missingness Assumption~\ref{ass:MCAR}. 
To investigate this, we conduct the same simulations as in the main paper, also for an MAR missingness mechanism. In doing so, the missing probabilities are adapted in a way that the first measurement always exists, and therefore $\lambda_1.=n$. This can be seen as a sort of baseline value which is always observable before any intervention.
Next, we let the other missingness probabilities depend on the performance in this initial measurement:
let $\boldsymbol{\Lambda} = (\Lambda_1, \dots, \Lambda_d) = (P(\lambda_1 = 0),\dots,P( \lambda_d=0))$  denote the vector of component-specific missingness probabilities.
 Using this, we define
\textcolor{black}{
\[P(\lambda_{ik}=1\mid X_{1k})= \begin{cases} 1, & i=1,\\ 1-\Lambda_i, & i\geq 2 \text{ and } X_{1k}\geq 0.75,\\ 1-(\Lambda_i+0.05), & i\geq 2 \text{ and } X_{1k}<0.75. \end{cases}\]}
for $i=1,...,d$ and $k=1,...,n$,

where again $\boldsymbol{\Lambda}=(0.1, 0.1, 0.12, 0.19, 0.13, 0.18, 0.08, 0.1, 0.19, 0.11, 0.13)$ in the one-way layout and $\boldsymbol{\Lambda}=(0.1, 0.12, 0.19, 0.12, 0.08, 0.18)$ in the two way layout.
Note that, formally, we have always split the component index into two indices in the two-way layout. 
For simplicity, we avoid this in the present section so as not to overload the notation.

This represents the idea that pupils with good results are more likely to be present during further tests. The observations $\vX$ remain unchanged, as in Section~\ref{simulation}.

The results are displayed in Tables~\ref{tab:simulation_mar_11x1_tests_rows}--\ref{tab:simulation_mar_3x2_ab_tests_rows} and illustrated in Figures~\ref{fig:MAR_size_n}--\ref{fig:MAR_power}.
The tests seem to be slightly more liberal than under the MCAR assumption.
However, the results are overall quite similar to those under the MCAR regime; the quasi-randomization tests appear quite robust against certain deviations from the MCAR assumption.
Thus, no additional description of the results is provided.

\begin{table}[h]
\caption{Empirical rejection rates in percent for the one-factor setting with $d=11$ under the MAR missingness mechanism with component-specific missingness probabilities. The hypothesis $H_{0,\textnormal{equal}}$ is tested at the nominal significance level $\alpha=5\%$. The case $\delta=0$ corresponds to the null hypothesis, whereas $\delta>0$ represents a trend alternative.}
\label{tab:simulation_mar_11x1_tests_rows}
\centering
\small
\setlength{\tabcolsep}{3.5pt}
\renewcommand{\arraystretch}{1.05}

\begin{tabular}{@{}lll*{9}{r}@{}}
\hline
&&&\multicolumn{3}{c}{$n=20$} & \multicolumn{3}{c}{$n=35$} & \multicolumn{3}{c}{$n=50$} \\
\cline{4-6}
\cline{7-9}
\cline{10-12}
Copula & $\rho_X$ & Test & \multicolumn{3}{c}{$\delta$} & \multicolumn{3}{c}{$\delta$} & \multicolumn{3}{c}{$\delta$} \\
\cline{4-6}
\cline{7-9}
\cline{10-12}
&&& 0.00 & 0.15 & 0.30 & 0.00 & 0.15 & 0.30 & 0.00 & 0.15 & 0.30 \\
\hline

Clayton & $0.1$ & $\Psi^\pi_n$ & \textbf{5.00} & 8.62 & 19.82 & \textbf{4.92} & 12.68 & 49.96 & \textbf{5.14} & 19.60 & 73.78 \\
 &  & $\Psi^\textnormal{boot}_n$ & \textbf{5.06} & 8.78 & 20.26 & \textbf{5.00} & 12.94 & 50.38 & \textbf{5.12} & 19.92 & 73.72 \\
 &  & $\Psi^\textnormal{asy}_n$ & 57.32 & 67.56 & 86.08 & 28.16 & 47.88 & 85.68 & 20.00 & 46.48 & 92.30 \\[4pt]
 & $0.2$ & $\Psi^\pi_n$ & \textbf{5.02} & 8.04 & 20.66 & \textbf{5.36} & 11.72 & 48.64 & \textbf{4.72} & 17.84 & 72.48 \\
 &  & $\Psi^\textnormal{boot}_n$ & \textbf{5.50} & 8.48 & 21.78 & 5.78 & 11.98 & 50.04 & \textbf{4.88} & 18.14 & 72.98 \\
 &  & $\Psi^\textnormal{asy}_n$ & 55.56 & 62.94 & 83.40 & 29.00 & 44.64 & 83.92 & 18.14 & 42.38 & 90.84 \\
\hline

Gumbel & $0.1$ & $\Psi^\pi_n$ & \textbf{5.34} & 9.60 & 25.10 & \textbf{5.14} & 16.34 & 53.04 & \textbf{5.28} & 21.02 & 76.10 \\
 &  & $\Psi^\textnormal{boot}_n$ & \textbf{5.60} & 9.92 & 24.98 & \textbf{5.18} & 16.20 & 53.16 & \textbf{5.16} & 21.12 & 76.22 \\
 &  & $\Psi^\textnormal{asy}_n$ & 58.14 & 67.96 & 87.66 & 28.52 & 51.78 & 88.24 & 19.80 & 48.24 & 93.18 \\[4pt]
 & $0.2$ & $\Psi^\pi_n$ & \textbf{5.48} & 11.72 & 29.00 & \textbf{5.42} & 16.60 & 57.60 & \textbf{5.52} & 24.64 & 78.12 \\
 &  & $\Psi^\textnormal{boot}_n$ & 5.80 & 12.22 & 29.62 & \textbf{5.52} & 17.10 & 58.40 & \textbf{5.50} & 25.04 & 78.36 \\
 &  & $\Psi^\textnormal{asy}_n$ & 56.52 & 67.56 & 87.60 & 28.82 & 50.62 & 88.26 & 19.90 & 49.68 & 92.62 \\
\hline
\end{tabular}

\end{table}

\begin{table}[p]
\caption{Empirical rejection rates in percent for the two-way layout  $d_A = 3$ levels for factor $A$ and $d_B = 2$ levels for factor $B$, (resulting in $d=d_Ad_B=6)$ under the MAR missingness mechanism with component-specific missingness probabilities. The hypothesis $H_{0,\textnormal{equal}}$ is tested at the nominal significance level $\alpha=5\%$. The case $\delta=0$ corresponds to the null hypothesis, whereas $\delta>0$ represents a one-component alternative.}
\label{tab:simulation_mar_3x2_equal_tests_rows}
\centering
\small
\setlength{\tabcolsep}{3.5pt}
\renewcommand{\arraystretch}{1.05}

\begin{tabular}{@{}lll*{9}{r}@{}}
\hline
&&&\multicolumn{3}{c}{$n=20$} & \multicolumn{3}{c}{$n=35$} & \multicolumn{3}{c}{$n=50$} \\
\cline{4-6}
\cline{7-9}
\cline{10-12}
Copula & $\rho_X$ & Test & \multicolumn{3}{c}{$\delta$} & \multicolumn{3}{c}{$\delta$} & \multicolumn{3}{c}{$\delta$} \\
\cline{4-6}
\cline{7-9}
\cline{10-12}
&&& 0.00 & 0.15 & 0.30 & 0.00 & 0.15 & 0.30 & 0.00 & 0.15 & 0.30 \\
\hline

Clayton & $0.1$ & $\Psi^\pi_n$ & \textbf{5.00} & 22.24 & 75.94 & \textbf{4.98} & 43.60 & 98.04 & \textbf{5.24} & 62.40 & 99.88 \\
 &  & $\Psi^\textnormal{boot}_n$ & 5.66 & 23.98 & 77.42 & \textbf{5.16} & 44.74 & 98.16 & \textbf{5.58} & 63.24 & 99.86 \\
 &  & $\Psi^\textnormal{asy}_n$ & 21.40 & 53.52 & 94.64 & 12.80 & 63.76 & 99.42 & 11.06 & 75.74 & 99.96 \\[4pt]
 & $0.2$ & $\Psi^\pi_n$ & \textbf{5.22} & 27.04 & 85.52 & \textbf{5.52} & 52.58 & 99.28 & \textbf{5.02} & 72.36 & 99.98 \\
 &  & $\Psi^\textnormal{boot}_n$ & 6.06 & 29.62 & 87.22 & 6.00 & 53.88 & 99.30 & \textbf{5.30} & 73.06 & 99.98 \\
 &  & $\Psi^\textnormal{asy}_n$ & 20.66 & 58.58 & 97.18 & 12.92 & 70.16 & 99.78 & 10.12 & 82.84 & 99.98 \\
\hline

Gumbel & $0.1$ & $\Psi^\pi_n$ & \textbf{5.12} & 22.12 & 76.84 & \textbf{4.58} & 43.56 & 97.78 & \textbf{5.46} & 62.40 & 99.88 \\
 &  & $\Psi^\textnormal{boot}_n$ & \textbf{5.56} & 23.62 & 78.76 & \textbf{4.82} & 44.28 & 98.04 & \textbf{5.56} & 62.74 & 99.88 \\
 &  & $\Psi^\textnormal{asy}_n$ & 21.84 & 53.72 & 95.52 & 12.86 & 64.28 & 99.48 & 10.36 & 75.56 & 100.00 \\[4pt]
 & $0.2$ & $\Psi^\pi_n$ & \textbf{4.58} & 26.26 & 86.24 & \textbf{4.98} & 51.64 & 99.40 & \textbf{4.88} & 71.66 & 100.00 \\
 &  & $\Psi^\textnormal{boot}_n$ & \textbf{5.12} & 27.92 & 87.74 & \textbf{5.34} & 52.80 & 99.42 & \textbf{5.26} & 72.00 & 100.00 \\
 &  & $\Psi^\textnormal{asy}_n$ & 19.80 & 56.98 & 97.42 & 12.18 & 68.72 & 99.82 & 10.30 & 81.60 & 100.00 \\
\hline
\end{tabular}

\end{table}

\begin{table}[p]
\caption{Empirical rejection rates in percent for the two-way layout with $d_A=3$ levels for factor $A$ and $d_B=2$ levels for factor $B$, (resulting in $d=d_Ad_B=6$) under the MAR missingness mechanism with component-specific missingness probabilities. The hypothesis $H_{0A}$ is tested at the nominal significance level $\alpha=5\%$. The case $\delta=0$ corresponds to the null hypothesis, whereas $\delta>0$ represents a one-component alternative.}
\label{tab:simulation_mar_3x2_a_tests_rows}
\centering
\small
\setlength{\tabcolsep}{3.5pt}
\renewcommand{\arraystretch}{1.05}

\begin{tabular}{@{}lll*{9}{r}@{}}
\hline
&&&\multicolumn{3}{c}{$n=20$} & \multicolumn{3}{c}{$n=35$} & \multicolumn{3}{c}{$n=50$} \\
\cline{4-6}
\cline{7-9}
\cline{10-12}
Copula & $\rho_X$ & Test & \multicolumn{3}{c}{$\delta$} & \multicolumn{3}{c}{$\delta$} & \multicolumn{3}{c}{$\delta$} \\
\cline{4-6}
\cline{7-9}
\cline{10-12}
&&& 0.00 & 0.15 & 0.30 & 0.00 & 0.15 & 0.30 & 0.00 & 0.15 & 0.30 \\
\hline

Clayton & $0.1$ & $\Psi^\pi_n$ & \textbf{5.00} & 41.64 & 95.40 & \textbf{4.68} & 67.16 & 99.94 & \textbf{5.50} & 84.50 & 100.00 \\
 &  & $\Psi^{\pi\min}_n$ & \textbf{5.04} & 41.40 & 95.12 & \textbf{4.64} & 67.18 & 99.94 & \textbf{5.40} & 84.76 & 100.00 \\
 &  & $\Psi^\textnormal{boot}_n$ & \textbf{5.50} & 42.44 & 95.66 & \textbf{4.84} & 67.44 & 99.94 & \textbf{5.38} & 84.86 & 100.00 \\
 &  & $\Psi^\textnormal{asy}_n$ & 9.76 & 55.16 & 97.72 & 7.18 & 73.82 & 99.98 & 7.08 & 87.84 & 100.00 \\[4pt]
 & $0.2$ & $\Psi^\pi_n$ & 5.78 & 47.06 & 97.50 & \textbf{5.50} & 74.68 & 100.00 & \textbf{5.12} & 89.76 & 100.00 \\
 &  & $\Psi^{\pi\min}_n$ & 5.88 & 46.68 & 97.26 & \textbf{5.52} & 74.20 & 100.00 & \textbf{5.06} & 89.82 & 100.00 \\
 &  & $\Psi^\textnormal{boot}_n$ & 6.24 & 48.52 & 97.72 & 5.74 & 75.02 & 100.00 & \textbf{5.18} & 90.18 & 100.00 \\
 &  & $\Psi^\textnormal{asy}_n$ & 9.86 & 59.50 & 98.92 & 7.92 & 79.88 & 100.00 & 6.66 & 91.88 & 100.00 \\
\hline

Gumbel & $0.1$ & $\Psi^\pi_n$ & \textbf{5.36} & 41.90 & 95.88 & \textbf{5.56} & 68.24 & 99.74 & \textbf{5.22} & 85.60 & 100.00 \\
 &  & $\Psi^{\pi\min}_n$ & \textbf{5.48} & 41.84 & 95.58 & \textbf{5.48} & 68.06 & 99.78 & \textbf{5.32} & 85.32 & 100.00 \\
 &  & $\Psi^\textnormal{boot}_n$ & 5.68 & 42.62 & 96.00 & 5.68 & 68.60 & 99.72 & \textbf{5.12} & 85.62 & 100.00 \\
 &  & $\Psi^\textnormal{asy}_n$ & 10.48 & 55.00 & 98.26 & 8.06 & 75.16 & 99.86 & 7.00 & 88.14 & 100.00 \\[4pt]
 & $0.2$ & $\Psi^\pi_n$ & 5.70 & 48.12 & 97.50 & 5.90 & 75.18 & 100.00 & \textbf{4.94} & 90.34 & 100.00 \\
 &  & $\Psi^{\pi\min}_n$ & 5.72 & 47.90 & 97.40 & 5.78 & 74.76 & 100.00 & \textbf{5.00} & 90.36 & 100.00 \\
 &  & $\Psi^\textnormal{boot}_n$ & 6.06 & 49.20 & 97.68 & 5.90 & 75.34 & 100.00 & \textbf{5.04} & 90.66 & 100.00 \\
 &  & $\Psi^\textnormal{asy}_n$ & 9.90 & 60.18 & 99.00 & 7.50 & 80.54 & 100.00 & 6.64 & 92.16 & 100.00 \\
\hline
\end{tabular}

\end{table}

\begin{table}[p]
\caption{Empirical rejection rates in percent for the two-way layout with $d_A=3$ levels for factor $A$ and $d_B=2$ levels for factor $B$, resulting in $d=d_Ad_B=6$, under the MAR missingness mechanism with component-specific missingness probabilities. The hypothesis $H_{0AB}$ is tested at the nominal significance level $\alpha=5\%$. The case $\delta=0$ corresponds to the null hypothesis, whereas $\delta>0$ represents a one-component alternative.}
\label{tab:simulation_mar_3x2_ab_tests_rows}
\centering
\small
\setlength{\tabcolsep}{3.5pt}
\renewcommand{\arraystretch}{1.05}

\begin{tabular}{@{}lll*{9}{r}@{}}
\hline
&&&\multicolumn{3}{c}{$n=20$} & \multicolumn{3}{c}{$n=35$} & \multicolumn{3}{c}{$n=50$} \\
\cline{4-6}
\cline{7-9}
\cline{10-12}
Copula & $\rho_X$ & Test & \multicolumn{3}{c}{$\delta$} & \multicolumn{3}{c}{$\delta$} & \multicolumn{3}{c}{$\delta$} \\
\cline{4-6}
\cline{7-9}
\cline{10-12}
&&& 0.00 & 0.15 & 0.30 & 0.00 & 0.15 & 0.30 & 0.00 & 0.15 & 0.30 \\
\hline

Clayton & $0.1$ & $\Psi^\pi_n$ & \textbf{4.94} & 14.92 & 45.54 & \textbf{5.14} & 24.12 & 74.64 & \textbf{4.76} & 35.32 & 89.70 \\
 &  & $\Psi^{\pi\min}_n$ & \textbf{4.96} & 14.74 & 45.80 & \textbf{5.12} & 24.26 & 74.38 & \textbf{4.92} & 34.98 & 89.82 \\
 &  & $\Psi^\textnormal{boot}_n$ & \textbf{5.36} & 15.20 & 46.54 & \textbf{5.04} & 24.30 & 74.46 & \textbf{4.86} & 35.20 & 89.80 \\
 &  & $\Psi^\textnormal{asy}_n$ & 9.74 & 23.78 & 59.30 & 7.86 & 30.14 & 80.46 & 6.58 & 39.70 & 92.30 \\[4pt]
 & $0.2$ & $\Psi^\pi_n$ & \textbf{4.94} & 17.70 & 52.80 & \textbf{4.94} & 28.50 & 81.60 & \textbf{5.28} & 39.34 & 93.46 \\
 &  & $\Psi^{\pi\min}_n$ & \textbf{5.00} & 17.72 & 53.04 & \textbf{4.92} & 28.28 & 81.92 & \textbf{5.16} & 39.30 & 93.70 \\
 &  & $\Psi^\textnormal{boot}_n$ & \textbf{5.34} & 18.28 & 54.06 & \textbf{5.14} & 29.04 & 82.06 & \textbf{5.22} & 39.58 & 93.54 \\
 &  & $\Psi^\textnormal{asy}_n$ & 9.18 & 26.08 & 65.78 & 7.24 & 34.54 & 86.24 & 6.68 & 44.06 & 95.24 \\
\hline

Gumbel & $0.1$ & $\Psi^\pi_n$ & \textbf{5.32} & 14.84 & 46.56 & \textbf{4.74} & 23.66 & 74.90 & \textbf{5.48} & 34.70 & 90.32 \\
 &  & $\Psi^{\pi\min}_n$ & \textbf{5.34} & 14.84 & 46.50 & \textbf{4.88} & 23.80 & 74.98 & \textbf{5.38} & 34.72 & 90.26 \\
 &  & $\Psi^\textnormal{boot}_n$ & \textbf{5.44} & 15.32 & 47.18 & \textbf{4.84} & 23.68 & 75.08 & \textbf{5.56} & 35.18 & 90.30 \\
 &  & $\Psi^\textnormal{asy}_n$ & 9.90 & 23.98 & 58.86 & 6.98 & 29.76 & 80.56 & 7.06 & 40.02 & 92.48 \\[4pt]
 & $0.2$ & $\Psi^\pi_n$ & \textbf{5.08} & 17.20 & 53.88 & \textbf{5.58} & 30.32 & 82.56 & \textbf{5.04} & 40.58 & 93.98 \\
 &  & $\Psi^{\pi\min}_n$ & \textbf{5.00} & 17.10 & 54.38 & \textbf{5.56} & 30.16 & 82.64 & \textbf{5.14} & 40.60 & 94.06 \\
 &  & $\Psi^\textnormal{boot}_n$ & \textbf{5.30} & 17.86 & 54.92 & 5.80 & 30.60 & 82.64 & \textbf{5.08} & 41.18 & 94.08 \\
 &  & $\Psi^\textnormal{asy}_n$ & 9.06 & 25.16 & 66.24 & 7.70 & 36.86 & 86.84 & 6.32 & 45.84 & 95.12 \\
\hline
\end{tabular}

\end{table}
\FloatBarrier

\color{blue}
\begin{figure}
\centering
\includegraphics[width=0.65\textwidth]{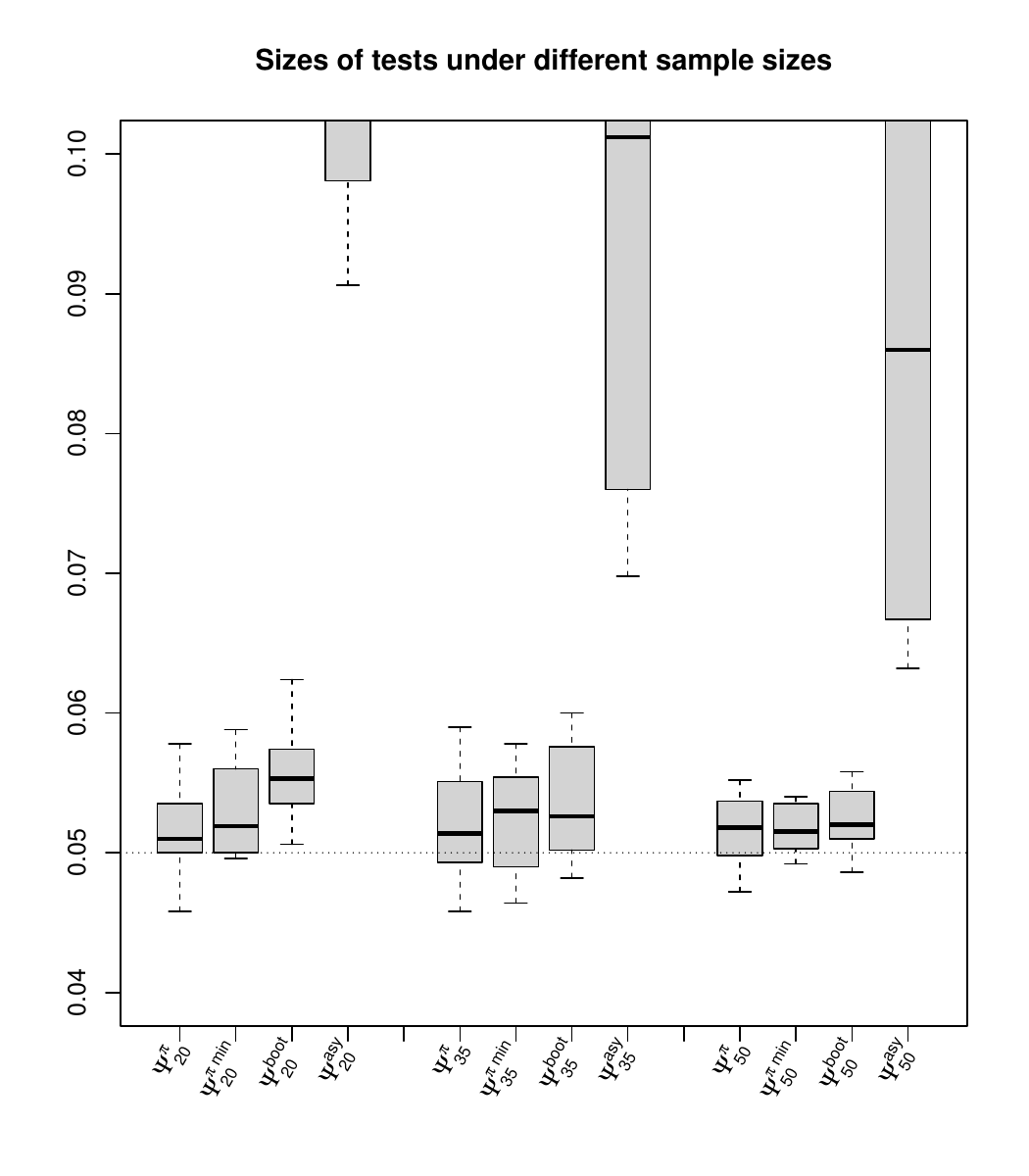}
\caption{Simulated sizes of all tests aggregated over all copulae and null hypotheses, under MAR. 
The dotted line indicates the nominal significance level $\alpha=5\%$. Note that  $\Psi_n^{\pi\textnormal{min}}$ is not available for testing all hypotheses, so comparisons with this test demands caution.}
\label{fig:MAR_size_n}
\end{figure}

\begin{figure}[h]
\centering
\includegraphics[width=0.65\textwidth]{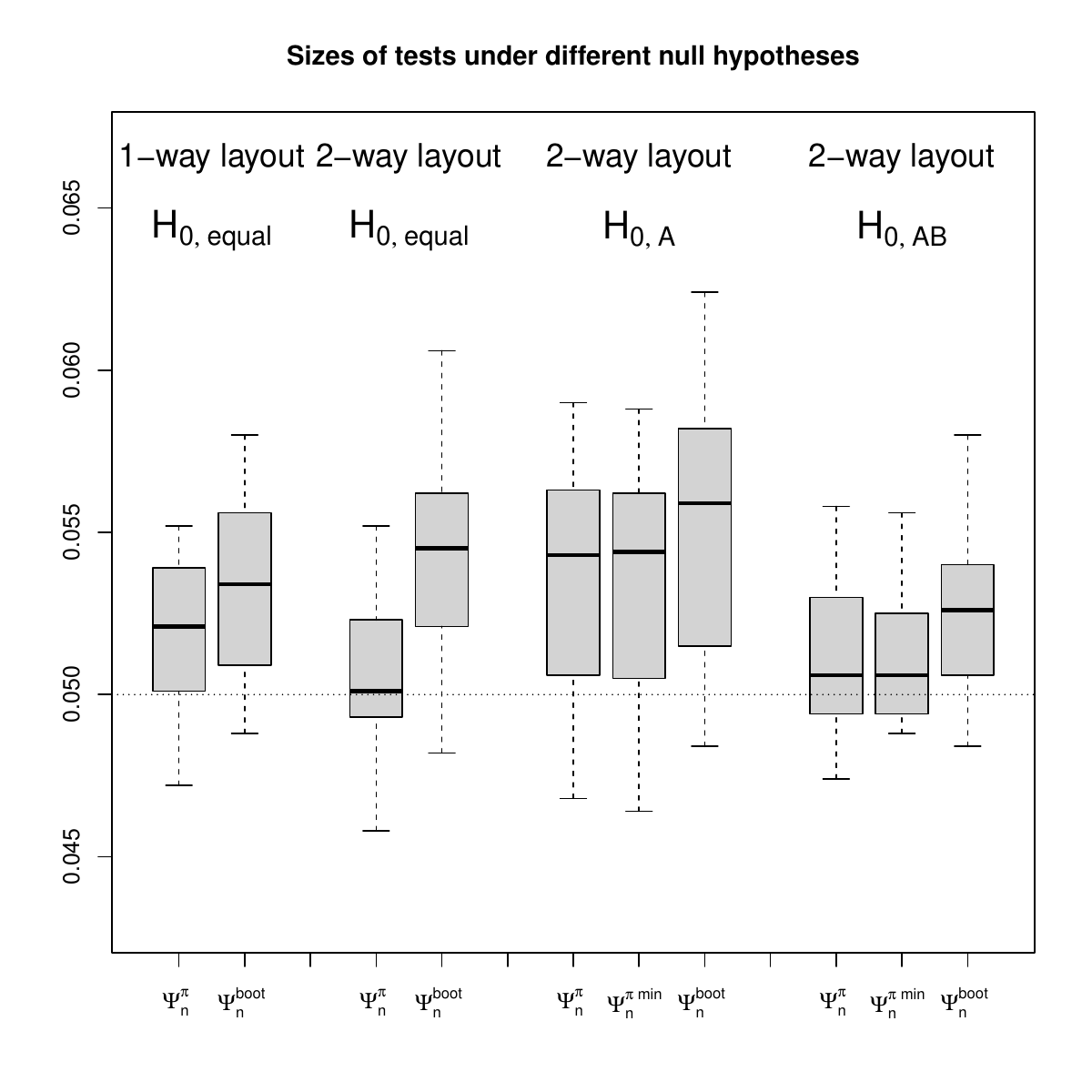}
\caption{Simulated sizes of all tests aggregated over all copulae and sample sizes, under MAR. 
The dotted line indicates the nominal significance level $\alpha=5\%$. Note that  $\Psi_n^{\pi\textnormal{min}}$ is not available for testing $H_{0,\textnormal{equal}}$.}
\label{fig:MAR_size_H0}
\end{figure}

\begin{figure}[h]
\centering
\includegraphics[width=0.45\textwidth]{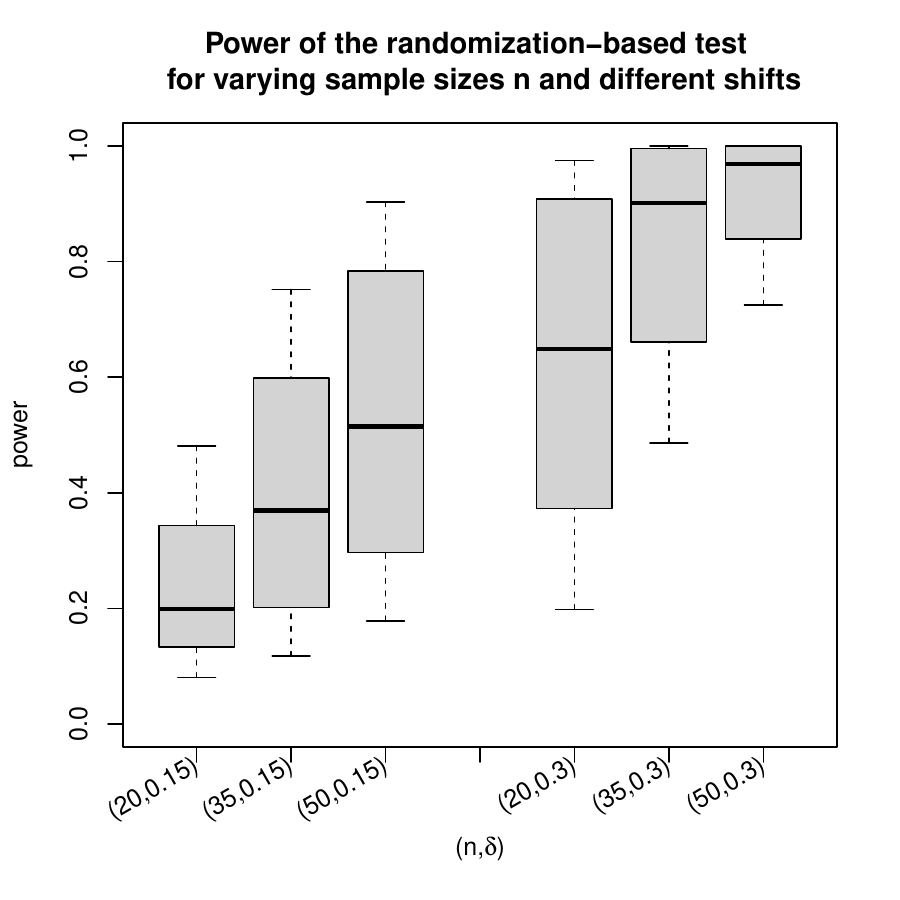}
\caption{Simulated power of the test $\Psi_n^\pi$ aggregated over all copulae, with varying sample sizes $n \in \{20,35,50\}$ (both of three adjacent boxplots) and shifts $\delta \in \{0.15, 0.3\}$ of the raw data, under MAR.}
\label{fig:MAR_power}
\end{figure}
\color{black}

\newpage
\clearpage
\newpage

\bibliographystyle{abbrvnat}
\bibliography{References}

\end{document}